\documentclass[reprint,amsmath,amsfonts,amssymb,aps,pra,superscriptaddress, nofootinbib]{revtex4-2} 

\usepackage{xcolor}
\usepackage{soul}
\usepackage{relsize} 
\usepackage{amsthm}  
\usepackage{hyperref} 
\usepackage{cleveref} 
\usepackage{centernot}
\usepackage{mathdots} 
\usepackage{amssymb} 
\usepackage{tikz}
\usetikzlibrary{arrows} 
\usetikzlibrary{positioning} 
\usepackage{qcircuit}
\usepackage[mathscr]{eucal}
\usepackage{comment}
\usepackage{enumitem} 
\usepackage[export]{adjustbox} 

\newtheorem{mytheorem}{Theorem}

\crefname{mytheorem}{theorem}{theorems}
\newtheorem*{butheorem}{Borsuk-Ulam theorem \cite{borsuk1933drei}}
\newtheorem{mycor}{Corollary}

\newtheorem{myconstr}{Construction}
\newtheorem{mylemma}{Lemma}

\theoremstyle{definition}
\newtheorem{mydef}{Definition}
\Crefname{mydef}{Definition}{Definitions} 
\theoremstyle{plain}

\newcommand{\bra}[1]{\left<#1\right|}
\newcommand{\ket}[1]{\left|#1\right>}

\usepackage[caption=false]{subfig} 

\usepackage{pifont} 
\usepackage{booktabs} 
\usepackage{pgfplots} 
\pgfplotsset{compat=1.17}
\usetikzlibrary{arrows.meta} 
\usetikzlibrary{math} 
\usetikzlibrary{calc} 
\usepackage{xcolor} 
\definecolor{blue0}{RGB}{130, 190, 210}
\definecolor{blue1}{RGB}{70, 150, 190}
\definecolor{blue2}{RGB}{60, 100, 180}
\definecolor{blue3}{RGB}{50, 63, 170}
\definecolor{blue4}{RGB}{20, 30, 100}
\definecolor{blue5}{RGB}{0, 0, 75}
\tikzset{%
  dashes/.style args={#1}{%
    dash pattern=on #1 off #1
  }
}

\usepackage{url}

\newcommand{\nocontentsline}[3]{}
\newcommand{\tocless}[2]{\bgroup\let\addcontentsline=\nocontentsline#1{#2}\egroup}

\begin{document}

\title{Topological obstructions to quantum computation with unitary oracles}
\author{Zuzana Gavorov\'a}
\affiliation{Rachel and Selim Benin
School of Computer Science and Engineering,
The Hebrew University of Jerusalem,
91904 Jerusalem, Israel}
\affiliation{F\'isica Teòrica: Informaci\'o i Fen\`omens Qu\`antics, Departament de F\'isica, Universitat Aut\`onoma de Barcelona, 08193 Bellaterra (Barcelona), Spain}
\author{Matan Seidel}
\affiliation{School of Mathematical Sciences, Tel Aviv University, 69978 Tel Aviv, Israel}
\author{Yonathan Touati}
\affiliation{Rachel and Selim Benin
School of Computer Science and Engineering,
The Hebrew University of Jerusalem,
91904 Jerusalem, Israel}


\begin{abstract}
Algorithms with unitary oracles can be nested, which makes them extremely versatile. An example is the phase estimation algorithm used in many candidate algorithms for quantum speed-up.
The search for new quantum algorithms benefits from understanding their limitations: Some tasks are impossible in quantum circuits, although their classical versions are easy -- for example cloning.
An example with a unitary oracle $U$ is the if clause, the task to implement controlled $U$ (up to the phase on $U$). In classical computation the conditional statement is easy and essential. In quantum circuits the if clause was shown impossible from one query to $U$. Is it possible from polynomially many queries?
Here we unify algorithms with a unitary oracle and develop a topological method to prove their limitations: 
No number of queries to $U$ and $U^\dagger$ lets quantum circuits implement the if clause, even if admitting approximations, postselection and relaxed causality.
We also show limitations of process tomography, oracle neutralization, and $\sqrt[\dim U]{U}$, $U^T$, and $U^\dagger$ algorithms. 
Our results strengthen an advantage of linear optics, challenge the experiments on relaxed causality, and motivate new algorithms with many-outcome measurements.
\end{abstract}

\maketitle 

\section{Introduction}

Quantum computation is believed to outperform classical computation. Many candidate algorithms for quantum speedup  \cite{shor1994algorithms,harrow2009quantum,brandao2017quantum, van2020quantum,temme2011quantum, whitfield2011simulation} use Kitaev's phase estimation algorithm \cite{kitaev95}. Phase estimation is so versatile because it is an \emph{oracle} algorithm; it estimates the phase of any unitary oracle $U$. In oracle algorithms, the input is an oracle -- an unknown subroutine, a black box whose inner workings are hidden. It is only possible to \emph{query} an oracle; to set its input and receive its output. With no assumptions on the subroutines, versatility is the main advantage of oracle algorithms. They can be gradually composed into complicated computation in independent building blocks.

Another advantage of oracles is the provability of limitations of computation. Factoring has a low-complexity quantum algorithm due to Shor \cite{shor1994algorithms}, while it is \emph{believed} to have only high-complexity classical algorithms. Rigorously proving complexity lower bounds is hard. Unlike factoring, \emph{oracle} problems do have some lower bound proof methods: the polynomial \cite{lbounds_polynomials}, hybrid \cite{bennett1997strengths} and adversary \cite{ambainis2000quantum,hoyer2007negative} methods. For example, the problem solved by Grover's quantum algorithm \cite{grover1996fast} is an oracle problem. It has a \emph{proven} complexity lower bound in classical  \cite{grover1996fast} as well as quantum \cite{bennett1997strengths,lbounds_polynomials} circuits. In Grover's quantum algorithm, the oracle is classical; it is a function on bit strings (representing classical computation). 
In phase estimation \cite{kitaev95} the oracle is quantum; representing any quantum circuit, it is a unitary matrix. Additional versatile algorithms could be singled out by a new -- possibly purely quantum -- lower-bound method.

Discovering limitations of quantum computation has further advantages. Consider the classically easy tasks, whose quantum analogs are impossible: cloning  \cite{no_cloning}, deleting  \cite{no_deleting}, universal-NOT gate  \cite{buzek1999unot}, programming  \cite{nielsen1997programmable}. These impossibilities lead to new protocols and deepen our understanding of physics; for example, no cloning inspires quantum cryptography \cite{gisin2002quantum} and serves as a check for various axiomatizations of quantum mechanics  \cite{chiribella2010purification,coecke2011categorical,kent2012minkowski,zurek2009darwinism,vitanyi2001kolmogorov}.

Usually, after an initial impossibility result, stronger no-go theorems and precise bounds follow by studying approximation, postselection, and complexity generalizations. The first two generalizations are often considered separately; approximate cloning \cite{buvzek1996quantum,buvzek1997quantum} independently from postselection (probabilistic) cloning \cite{duan1998probabilistic}, and similarly for other tasks \cite{hillery2006approximate, yang2020optimal, nielsen1997programmable, vidal2002storing, sedlak2019optimal, ishizaka2009quantum} (see \cite{chefles1999strategies} for an exception). 
The third generalization allows repeated input. Re-preparing a state or re-running a procedure is often necessary, for example for obtaining statistics from experiments. Hopefully, polynomially many repetitions suffice. One such complexity generalization is $N$-to-$M$ cloning \cite{gisin1997optimal, bruss1998optimal, chefles1999strategies, keyl1999optimal, werner1998optimal}: if $M=N+1$, $N$ is the \emph{sample complexity} of creating one clone. When the input is an oracle, we are interested in the algorithm's \emph{query complexity}.

Here we develop a unified formalism for quantum computation with unitary oracles, which captures all of the above generalizations simultaneously. Thus, we are able to derive the strongest-yet no-go theorem for the \emph{if-clause} task \cite{araujo_cU}; the task to implement controlled $U$ given an oracle $U$, the quantum version of the conditional statement essential to classical computation. We show that an if clause \emph{of any query complexity} is impossible, thus losing what could be a fundamental building block in quantum circuits. This, for example, limits the flexibility of phase estimation, which \emph{assumes} the controlled $U$ building blocks. Our method exploits features unique to unitary oracles; its novel topological approach adds to the few previously known lower-bound methods. To further demonstrate its applicability, we prove limitations of other previously studied unitary-oracle tasks: the neutralization, fractional power, inverse and transpose tasks. The limitations hold for quantum circuits (even with postselection and relaxed causality), but not for linear optics. This motivates developing implementation-dependent algorithms for oracle problems. 

\subsection{Problem statement}

To state the problem precisely, we first discus some preceding works. In addition to phase estimation \cite{kitaev95}, several known algorithms use oracle access to $U$ for implementing a function of $U$: the complex conjugate $\overline{U}$ \cite{compl_conj}, the transpose $U^T$ and inverse $U^\dagger$ \cite{quintino_inverse, quintino_transforming}, and the fractional power $U^q$ for some $q\in\mathbb{Q}$ \cite{sheridan_maslov_mosca, gilyen2019quantum}. The fractional power algorithms differ from the others in two ways: First, they require queries to $U$ (or controlled $U$) \emph{and} to its inverse. 
Second, they succeed not for \emph{all}, but only for \emph{most} $n$-qubit unitary oracles $U\in U(2^n)$ -- they are (exponential 
\footnote{The algorithms work for inputs $U$ that satisfy a condition captured by the gap parameter $g$ \protect{\cite{sheridan_maslov_mosca}}. As a function of $g$, the algorithms require $O(1/g)$ queries to succeed on these \protect{“good"} inputs. Haar measure over $U(2^n)$ induces the probability of a \protect{“bad"} input $\delta=1-(1-g)^{2^n}\leq 2^n g$, thus the algorithms fit within a requested $\delta$ with $O(2^n/\delta)$ queries. Since $\delta\geq 1-e^{-2^n g}$, the exponential query complexity is necessary.
}) average-case algorithms in the sense of \emph{heuristic schemes} \cite[Definition 11]{bogdanov2006average}.
Complexity theory mostly studies worst-case algorithms; algorithms succeeding even on the worst input. Worst-case algorithms satisfy input-independent guarantees, which is preferred for reliable computation. 

In different computational models the same oracle problem might have a different \emph{query complexity} -- the minimal number of queries the model needs to solve the problem. In the quantum-circuit model the queries have predefined causal order, and in process matrices the causal order is relaxed \cite{chiribella2013quantum,oreshkov2012quantum}. 
While quantum circuits need three queries, process matrices provide a two-query solution to the quantum switch \cite{chiribella2013quantum}, the task to superpose orders of two oracles. Thus, theoretically, process matrices are more powerful. But are they physically relevant? The debate about the physicality of the process-matrix solution to the quantum switch is ongoing, for example whether it was implemented \cite{oreshkov2019time} or only simulated \cite{paunkovic2020causal} by certain linear optics experiments \cite{procopio2015experimental,goswami2018indefinite,rubino2017experimental}. 

\begin{figure}
     \subfloat[]{
         \includegraphics[valign=c]{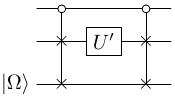}
         \vphantom{\includegraphics[valign=c]{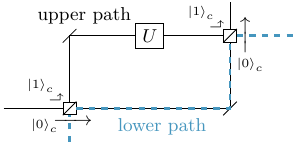}}
         \label{fig_kitaev_vacuum}
     }
     \subfloat[]{
         \includegraphics[valign=c]{fig1b_mach-zehnder_vacuum-cropped.pdf}
         \label{fig_mach-zehnder}
     }
        \caption{(a) This algorithm of Kitaev \protect{\cite{kitaev95}} implements $c_0(U')$ if given $\ket{\Omega}$ such that $U'\ket{\Omega}=\ket{\Omega}$. (b) The interferometric if clause \protect{\cite{araujo_cU}}: polarising beam splitters transmit or reflect the photon depending on its polarization degree of freedom $\left|i\right>_c$. $U$ acts on other degrees of freedom. 
}
        \label{fig:two graphs}
\end{figure}

A discrepancy between quantum circuits and linear optics is already exhibited by a simpler unitary-oracle problem -- the \emph{if clause}. 
Its classical version is fundamental for classical computation: Apply a subroutine if a control bit is $1$ and do nothing otherwise. In the quantum version, the control bit is a qubit and the subroutine is a $d$-dimensional unitary $U\in U(d)$, acting on $n=\lceil \log_2 d \rceil$ qubits. Following \cite{araujo_cU}, we identify this with the linear operator
\begin{equation}
    c_\phi(U):=\left|0\right>\!\!\left<0\right|\otimes I + e^{i \phi(U)}\left|1\right>\!\!\left<1\right|\otimes U\label{eq_if-clause}
\end{equation}
for some real function $\phi$. In the special case when $\phi$ is identically zero, the operator, $c_0(U)=\left|0\right>\!\!\left<0\right|\otimes I + \left|1\right>\!\!\left<1\right|\otimes U$, is controlled $U$. Access to $c_0(U)$ is assumed in phase estimation  \cite{kitaev95}. $c_0(U)$ is implementable given a $+1$ eigenstate (Fig. \ref{fig_kitaev_vacuum}) or a classical description \cite{barenco1995elementary, nielsen2002quantum} of $U$. Such information is not available in our oracle setting, which entertains more flexibility. When $U$ is an oracle, the freedom in $\phi(U)$ is necessary: It grants the if clause an insensitivity to $U$'s global phase  \cite{araujo_cU,soeda2013limitations,thompson2018quantum}. Ara{\'u}jo, Feix, Costa, and Brukner \cite{araujo_cU} proved that implementing \eqref{eq_if-clause} for all $U\in U(d)$ from only one query to $U$ is impossible in quantum circuits, though it is possible in linear optics (Fig. \ref{fig_mach-zehnder}). To explain the discrepancy, the authors argued that the gate $U$ in Fig. \ref{fig_mach-zehnder} is not completely unknown: its position is known, revealing that on the lower path-modes it acts as the identity. Any physical gate is restricted to a specific space (specific modes), so ref. \cite{araujo_cU} suggested adding the direct sum composition $U\oplus I$ to the quantum-circuit model. 

Dong, Nakayama, Soeda and Murao \cite{dong_controlled} found an algorithm for a related task, of implementing $c^d_\phi(U)$ for all $U\in U(d)$, where
 \begin{equation}
    c^m_\phi(U):=\left|0\right>\!\!\left<0\right|\otimes I+e^{i\phi(U)}\left|1\right>\!\!\left<1\right|\otimes U^m\text{.}\label{eq_cUm}
\end{equation}

\noindent In their algorithm $m=d$, $\phi(U)=\det(U)^{-1}$, and $U$ is queried $d$ times. Alternatively, with $d$ queries to the fractional power $U^\frac{1}{d}$ the algorithm implements $c^1_\phi(U)=c_\phi(U)$ ($m=1$), the if clause. Thus, the algorithm of \cite{dong_controlled} relates two unitary-oracle problems: the if clause and the fractional power.

The fractional power dates back to the $\smash{\sqrt{U}}$ question in Aaronson's 2006 list of “The ten most annoying questions in quantum computing” \cite{aaronsonten}. For a given fraction, a given unitary has several possible fractional powers. Sheridan, Maslov and Mosca \cite{sheridan_maslov_mosca} chose a particular fractional power, and, by finding a contradiction, proved that a worst-case algorithm implementing it for an unknown $U$ is impossible \cite[Lemma 1]{sheridan_maslov_mosca}. This oracle impossibility means that if a worst-case algorithm for some $\smash{\sqrt{U}}$ exists, it depends on the details of the implementation of $U$.

Another important oracle task is to remove the effect of an unknown and undesired evolution -- the \emph{neutralization} task. Often the evolution, the repeated application of some unitary $U$, cannot be avoided. However, initializing the state and/or interlacing $U$ with fixed gates can cancel out its effect. For example, spin echo \cite{spin_echo} is a neutralization algorithm for restricted $U$s; for qubit rotations around the $z$ axis. Spin echo neutralizes two queries to $U$. The neutralization algorithm of \cite{dong_controlled} works for general $d$-dimensional unitaries. It neutralizes $d$ queries to $U$.

The tasks mentioned above are all tasks with unitary oracles. Some algorithms and no-go theorems have been found, but maybe a unified approach is possible, with more emphasis on query complexity in order to strengthen the no-gos and to understand the optimality of the known algorithms. We have seen that the linear-optics complexity of the if clause is constant, $1$, and that its quantum-circuit complexity is strictly bigger. To bridge the gap, ref. \cite{araujo_cU} suggested modifying the quantum circuit formalism. But is the gap large enough -- e.g. beyond $\operatorname{poly}(n)$ \footnote{In the absence of oracles, the quantum-circuit model can simulate optics with only a polynomial overhead \protect{\cite{feynman1982simulating,abrams1997simulation,bravyi2002fermionic,aaronson2011computational}}.} -- to justify this? What is the quantum-circuit complexity of the if clause? It seems at most exponential, because $U$'s classical description suffices to build an if clause \cite{nielsen2002quantum}, and process tomography \cite{chuang1997prescription,poyatos1997complete,leung2000towards} yields the oracle's classical description from $\exp(n)$ queries.

\subsection{Results overview}

\begin{table*}
\caption{Examples of tasks and their approximate solvability by oracle algorithms with postselection}
\label{table_tasks2}
\begin{ruledtabular}
\begin{tabular}{l l l l l } 
  & task & for all $U\in U(d)$ & for most $U\in U(d)$ \\
 \hline
if clause & $( c_\phi, \{id, inv\})$ & \textbf{impossible} \footnotemark[1] & unitary algorithm \cite{dong_controlled,sheridan_maslov_mosca}\footnotemark[2] \\
if clause from fractions & $( c_\phi, \{U\mapsto U^\frac{1}{d}\})$ & exact unitary algorithm \cite{dong_controlled} & \\
complex conjugate & $(U\mapsto \overline{U}, \{id\})$ & exact unitary algorithm \cite{compl_conj} & \\
transpose & $(U\mapsto U^T, \{id\})$ & exact algorithm \cite{quintino_transforming, quintino_inverse} & \\
inverse & $(inv, \{id\})$ & exact algorithm \cite{quintino_transforming, quintino_inverse,compl_conj}\footnotemark[2] & \\
phase estimation & $(\operatorname{phase\_est}, \{c_0\})$ & unitary algorithm \cite{kitaev95} & \\
fractional power & $(U\mapsto U^\frac{1}{d}, \{id, inv\})$ & \textbf{impossible} \cite{sheridan_maslov_mosca}\footnotemark[1]\footnotemark[2] & unitary algorithm \cite{sheridan_maslov_mosca} 
\end{tabular}
\end{ruledtabular}
\footnotetext[1]{this result}
\footnotetext[2]{by composing two tasks}
\end{table*}

In this work we develop a unified formalism for computation with unitary oracles $U$. We phrase such tasks and worst-case algorithms as functions on unitaries. The algorithm functions have two important properties: continuity and homogeneity. We add to the few previously known lower-bound methods a different, topological method and derive results.

The first and the most surprising result shows that the quantum-circuit complexity of the if clause is infinite. We show that no matter how many times a postselection circuit queries $U$ and $U^\dagger$, it cannot implement the if clause with a nonzero success probability for every $U\in U(d)$ -- not even approximately! Surprisingly, process tomography suggested above fails for the if clause, and only works for a relaxed variant of the task. We prove this limitation of process tomography directly. Second, we give a different proof that implementing the fractional power $U^\frac{1}{d}$ for all $U\in U(d)$ is impossible, from any number of queries to $U$ and $U^\dagger$. Our contradiction is independent of the particular $d$th root function. Third, we show that quantum circuits and process matrices fail to neutralize some specific numbers of queries. Last, we prove related limitations regarding the transpose and inverse tasks.

The above results limit versatile quantum computation, and impact our understanding of tomography, measurements, linear optics and causality. Using process tomography for the if clause has a caveat: Instead of a superoperator estimate of $\rho\mapsto U\rho\, U^\dagger$, the if clause requires a matrix estimate of $U$. We show the limitation of such matrix tomography. Defining a relaxed if clause circumvents the limitation, but the algorithm must use a measurement beyond the binary success/fail type. This splits measurements into two groups with different effects on the quantum-circuit query complexity. The quantum-circuit model itself is compared to other models: linear optics and process matrices. On the if clause the quantum-circuit model turns out to be infinitely less efficient than linear optics! One might put some hope into relaxing causality; maybe linear optics is better matched by process matrices. While arguably true for the quantum switch task, this is wrong for the more fundamental if-clause task: its process-matrix complexity is infinite, too. The advantage of linear optics stems from restricting the oracle. The models with fully general unitary oracles have a property central to our impossibility proofs: homogeneity. Linear optics restrict oracles to the form $1\oplus U$ which breaks homogeneity. Differently from \cite{araujo_cU}, we attribute the direct sum to the linearity of linear optics (see \Cref{section_discussion}).

The rest of the paper is organized as follows. Section \ref{section_def_postselalg} defines oracle computation using functions on $d$-dimensional unitaries, $U\in U(d)$. Section \ref{section_if-clause} proves the if-clause impossibility and the process tomography limitation by exploiting the continuity of algorithms and the topology of the space $U(d)$ (\Cref{thm_top}). Section \ref{section_tasks} uses this topological approach to prove results regarding the neutralization, $1/d$th power, transpose, and inverse. In Section \ref{section_discussion} we emphasize that our method applies to the worst-case models with the exception of linear optics. We discuss the cause and the significance of this exception. Then we discuss relaxed causality and measurements.

\section{Algorithm as a function of the oracle}\label{section_def_postselalg}

An algorithm should solve a problem. Each problem mentioned in the Introduction asks to implement some operator $t(U)\in L(\mathcal{H})$, where $L(\mathcal{H})$ is the set of linear operators from a finite-dimensional Hilbert space $\mathcal{H}$ to itself. We call the function $t:U(d)\to L(\mathcal{H})$ \emph{task function}. The allowed types of queries to the oracle, for example $U$ and $U^\dagger$ queries, are also represented by functions on operators (\emph{query functions}) in our example by $id:U\mapsto U$ and $inv:U\mapsto U^\dagger$. Another example is the $c_0$ query function in phase estimation. A set of allowed query functions is a \emph{query alphabet} $\Sigma$. A task function and a query alphabet together form a pair $(t, \Sigma)$ called \emph{task}. Examples are in \Cref{table_tasks2}, which phrases the above mentioned preceding works in terms of \emph{tasks}. 

Next we represent algorithms by functions of $U$. In \cite{lbounds_polynomials, oracles_survey} a unitary quantum circuit with $N$ queries to an oracle $U$ corresponds to a sequence of unitary transformations
\begin{equation*}
V_0, U, V_1, U, \dots V_{{N}-1}, U, V_{N}
\end{equation*}

\noindent for some fixed unitaries $V_i$ acting on a possibly larger Hilbert space than the oracle $U$. We generalize the above in two ways. First, we add the possibility to query $U$ via query functions $\sigma_i$:
\begin{equation*}
V_0, \sigma_1(U), V_1, \sigma_2(U), \dots V_{{N}-1}, \sigma_{N}(U), V_{N}
\end{equation*}

\noindent Regarding $\sigma_i$ as a single character, the string $\mathbf{s}=\sigma_1 \sigma_2 \dots\sigma_N$, called the \emph{query sequence}, is fixed for the algorithm. Second, we allow for a projective measurement at the end $\{\Pi_\text{succ},\Pi_\text{fail}\}$, followed by the postselection on success $\Pi_\text{succ}$.

\begin{mydef}[Postselection oracle algorithm]\label{def_postsel_alg} A \emph{postselection oracle algorithm} $(A, \Sigma_A)$ is a function $A:U(d)\to L(\mathcal{H}\otimes\mathcal{K})$ of the form
\begin{equation}\label{eq_general_alg}
    A(U) = \Pi_\text{succ}\, V_{N} (\sigma_N(U)\otimes I_{ \mathcal{K}_N}) \dots V_1(\sigma_1(U)\otimes I_{ \mathcal{K}_1})V_0 \text{,}
\end{equation}
where $\sigma_i\in \Sigma_A$ for all $i\in [N]$. The algorithm implements the total linear operator $A(U)$ (Fig. \ref{supp_fig_model}) whenever its projective measurement $\{\Pi_\text{succ},\Pi_\text{fail}\}$ yields success. The probability of success must be nonzero if the ancilla Hilbert space $\mathcal{K}$ is initialized to the all-zero state, i.e. for all $U\in U(d)$ and for all $\ket{\xi}\in\mathcal{H}$
\begin{equation}
    ||A(U)\left(\ket{\xi}\otimes\ket{\mathbf{0}}_ \mathcal{K}\right)||^2> 0\text{,}
    \label{condition_post}
\end{equation}
which is the usual postselection condition \cite{aaronson_postBQP}.
\end{mydef}

\begin{figure}[b!]
        \vspace{-10pt}
         \includegraphics[]{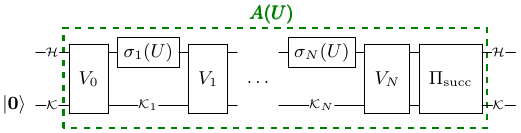}
         \caption{A postselection oracle algorithm has fixed unitary gates $V_i$, a fixed projection $\Pi_\text{succ}$ and special slots for queries $\sigma_i(U)$, 
         while identity is applied to the remaining Hilbert space $\mathcal{K}_i$.
         }
         \label{supp_fig_model}
\end{figure}

We will equivalently represent a postselection oracle algorithm by $(\mathcal{A}, \Sigma_A)$, where the calligraphic font denotes the corresponding function into \emph{superoperators}, i.e.
\begin{equation}
    \mathcal{A}(U)(\rho):=\operatorname{tr}_{\mathcal{K}}\left[A(U)\left(\rho\otimes\ket{\mathbf{0}}\!\!\bra{\mathbf{0}}_\mathcal{K}\right)A(U)^\dagger\right]\label{eq_postsel_superop}\text{,}
\end{equation}
\noindent where we already include the correct initialization of the ancilla, and its tracing-out at the end of the algorithm.
Denote by $\mathcal{S}_+(\mathcal{H})$ the completely positive superoperators $L(\mathcal{H})\to L(\mathcal{H})$ that are trace-nonvanishing, i.e. they map density states $\mathcal{D}(\mathcal{H})\subset L(\mathcal{H})$ to linear operators with strictly positive trace. Then the postselection condition \eqref{condition_post} translates to the requirement that $\mathcal{A}:U(d)\to \mathcal{S}_+(\mathcal{H})$.

The following generalization of postselection oracle algorithms relaxes the causal order of the queries.

\begin{mydef}[Process-matrix algorithm]\label{def_supermap}
A {\it process-matrix algorithm} $(\mathcal{A}, \Sigma_A)$ is a function $\mathcal{A}:U(d)\to \mathcal{S}_+(\mathcal{H})$ such that
\begin{equation}
    J_{\mathcal{A}(U)}=\operatorname{tr}_{{\substack{\text{inner}\\ \text{edges}}}}[W(J_{\tilde{\sigma}_1(U)}\otimes \dots \otimes J_{\tilde{\sigma}_N(U)}\otimes I_{{\substack{\text{open}\\ \text{edges}}}})]\label{eq:process_matrix}
\end{equation}
is the contraction in Fig. \ref{fig:process_matrix}, where $W$ is a fixed matrix,
$J_\Psi:=\sum_{i, j}\ket{i}\!\!\bra{j}\otimes\Psi(\ket{i}\!\!\bra{j})$ is the Choi isomorphism of $\Psi$,  and $\tilde{\sigma}_i(U)$ are query superoperators, $\tilde{\sigma}_i(U)(\rho):=\sigma_i(U)\rho\,\sigma_i(U)^\dagger$ for the query functions $\sigma_i\in\Sigma_A$.
\end{mydef}

\begin{figure}
    \centering
     \includegraphics[]{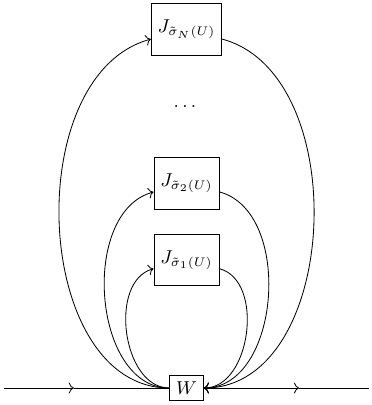}
     \caption{A process-matrix algorithm is a contraction of queries with a fixed matrix $W$. $W$ may or may not determine the queries' order.}
     \label{fig:process_matrix}
\end{figure}

\noindent As opposed to quantum circuits, the wires in Fig. \ref{fig:process_matrix} do not fix the order of the queries. How the queries are connected is determined by $W$, which may correspond to a fixed causal order, but also to a superposition of orders \cite{chiribella2013quantum}, and beyond \cite{oreshkov2012quantum, baumeler2016space}.

It remains to define what it means for an algorithm to achieve a task. Note that we consider only unitary $t(U)$.

\begin{mydef}[exactly achieving a task]\label{def_exactly_achieving}
A postselection oracle algorithm $(A, \Sigma_A)$ \emph{exactly achieves} the task $(t, \Sigma)$ if $\Sigma_A\subseteq \Sigma$ and if it implements $t(U)$ to the Hilbert space $\mathcal{H}$, provided that the ancilla Hilbert space $\mathcal{K}$ was initialized to the all-zero state $\left|\mathbf{0}\right>$, i.e. the operator equation
\begin{equation}
    A(U)\left(I\otimes\left|\mathbf{0}\right>\!\!\left<\boldsymbol{0}\right|\right)=t(U)\otimes\left|g(U)\right>\!\!\left<\boldsymbol{0}\right|\label{eq_exactly-achieves}
\end{equation}
holds for some $\left|g(U)\right>\in\mathcal{K}$.
\end{mydef}

\noindent The freedom in the unnormalized state $\left|g(U)\right>$ allows for any global phase and any postselection probability. By the postselection condition \eqref{condition_post} $\left|g(U)\right>$ is never zero. 

To define achieving a task for \emph{both} postselection oracle algorithms and process-matrix algorithms, we use an equivalent equation in terms of (trace non-vanishing) superoperators $\mathcal{A}(U)\in\mathcal{S}_+(\mathcal{H})$:

\begin{equation}
    \frac{\mathcal{A}(U)(\rho)}{\operatorname{tr}\!\left[\mathcal{A}(U)(\rho)\right]}=t(U)\rho \,t(U)^\dagger
    \label{eq_0-achieve}
\end{equation}
One direction of the equivalence is immediate from \eqref{eq_postsel_superop}, the other follows from Theorem 2.3 of \cite{technical}. The theorem also relates the \emph{errors} in the operator and superoperator languages. Here we continue with superoperators.

\begin{mydef}[$\epsilon$-approximately achieving a task]\label{def_supermap_approx_achieving}
Algorithm $(\mathcal{A}, \Sigma_A)$ {\it $\epsilon$-approxi\-mately achieves} the task $(t, \Sigma)$ if $\Sigma_A\subseteq \Sigma$ and if the renormalized output of the algorithm is always close to the task output, i.e. for all $U\in U(d)$
\begin{equation}
    \sup_{\substack{\mathcal{H}'\\[0.3em]\rho\in\mathcal{D}(\mathcal{H}\otimes\mathcal{H}')\hspace{-1.7em}}}\left|\left|\frac{\mathcal{A}(U)\!\otimes\! \mathcal{I}(\rho)}{\operatorname{tr}\!\left[\mathcal{A}(U)\!\otimes \!\mathcal{I}(\rho)\right]}-(t(U)\otimes I)\rho \,(t(U)\otimes I)^\dagger\right|\right|_{\operatorname{tr}}\!\!\!\leq \epsilon\text{,}\label{eq_eps-achieve}
\end{equation}
where $||\cdot||_{\operatorname{tr}}$ is the trace norm, and $I\in L(\mathcal{H}')$ and $\mathcal{I}\in \mathcal{S}_+(\mathcal{H}')$ are the identity operator and superoperator.
\end{mydef}

The left-hand side of inequality \eqref{eq_eps-achieve} is \emph{postselection diamond distance} \cite{technical}. It simplifies to diamond distance if $\mathcal{A}(U)$ is trace-preserving. Our use of the postselection condition and of diamond distance -- a worst-case measure -- means that \emph{$\epsilon$-approximately achieving} is worst-case: Even with the worst inputs $\rho$ and $U$ the algorithm succeeds with nonzero probability and then satisfies the error bound. Thus, our formalism excludes the two algorithms in column four of \Cref{table_tasks2}, and includes all the algorithms in column three, where “exact” means $\epsilon=0$, and “unitary” means $\Pi_\text{succ}=I$. 

\section{If-clause impossibility via topology}\label{section_if-clause}

Algorithms as functions of $U$ (\cref{eq_postsel_superop,eq:process_matrix}) have useful properties. The postselection oracle algorithm, function \eqref{eq_postsel_superop}, is continuous whenever the query functions are (see \eqref{eq_general_alg}). So is the process-matrix algorithm, function \eqref{eq:process_matrix}, because Choi isomorphism and its inverse are continuous. Moreover, both functions are $0$-homogeneous according to the definition below.

\begin{mydef}[homogeneity]\label{def_homog}
    A function $f:X\to Y$ is $m$-homogeneous for some $m\in\mathbb{Z}$ iff $f(\lambda x)=\lambda^m f(x)$ for any $x\in X$ and any scalar $\lambda$ such that $\lambda x\in X$.
\end{mydef}

\noindent For example, the scalars are restricted to the unit circle (one-sphere) $\lambda\in S^1$ if the domain of $f$ is the space of unitaries $U(d)$. The space $SU(2)$ is homeomorphic to the three-sphere $S^3$ and the following prominent result in topology studies continuous functions on $n$-spheres $S^n$:

\begin{butheorem}
    For any $f:S^n\to\mathbb{R}^n$ continuous there exist antipodal points mapped to the same value, $f(x)=f(-x)$.
\end{butheorem}

\noindent While the Borsuk-Ulam theorem (Fig. \ref{fig_tides}) implies a special case (see \Cref{section_bu}), for the fully general if-clause impossibility we prove a topological lemma that also exploits homogeneity.

\begin{figure}
    \includegraphics[width=0.22\textwidth]{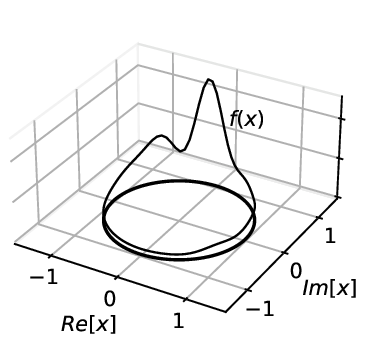}
    \vspace{-5pt}
    \includegraphics[width=0.22\textwidth]{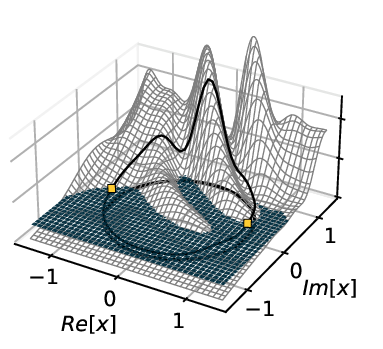}
    \caption{Intuition for the $n=1$ case of the Borsuk-Ulam theorem: At some moment the rising tide must hit points that are exactly opposite on the circle.}
    \label{fig_tides}
\end{figure}

\begin{mylemma}\label{thm_top}
Let $m\in\mathbb{Z}$. If
a continuous $m$-homogeneous function $f:U(d)\rightarrow S^{1}$ exists, then $d$ divides $m$.
\end{mylemma}

\noindent For example the determinant, emerging in the $c^d_\phi$ algorithm reviewed in \Cref{section_cUd}, is a $d$-homogeneous function: $\smash{\det(\lambda U)=\lambda^d \det(U)}$. By \Cref{thm_top} this homogeneity degree and its multiples are the only ones possible for continuous functions $U(d)\rightarrow S^{1}$. Consequently, in our main result $c^m_\phi$ obeys a dichotomy that depends on whether $d$ divides $m$ ($d|m$):

\begin{samepage}
\begin{mytheorem}\label{thm_main}
Let $d$ be the oracle dimension, let ${m}\in \mathbb{Z}$ and $\epsilon\in \left[0, 1/2\right)$. 

\setlength\topsep{0pt}
\begin{itemize}[topsep=0pt,align=left, leftmargin=5pt] 
\setlength\parskip{0pt}
\setlength\itemsep{1pt}
\item If $d | m$ a postselection oracle algorithm exists $\epsilon$-approximately achieving $(c^m_\phi, \{id, inv\})$
for some $\phi$. It makes $|m|$ queries.
\item If $d\nmid m$ 
no such postselection oracle algorithm or process-matrix algorithm exists, for any number of queries. 
\end{itemize}
\end{mytheorem}
\end{samepage}

The if-clause ($m=1$) impossibility is immediate. In addition to the following full proof of \Cref{thm_main}, the appendix contains two more proofs for only the exact, $\epsilon=0$, impossibility. The “operational” proof in \Cref{section_operational} reaches a contradiction by using the supposed $c^m_\phi$ algorithm as a building block in a larger circuit. The proof in \Cref{section_cover_spaces} for only the $m=1$ case \footnote{We thank an anonymous referee from the QIP conference for observing that such a proof is possible.} gives additional intuition: The special unitary group $SU(d)$ is a $d$th cover of $PU(d)$, the projective unitary group. This prevents the existence of a continuous map from a unitary superoperator to a matching operator, $PU(d)\to U(d)$, preventing an exact if-clause algorithm. The full proof below holds for approximations and relies on \Cref{thm_top} proven next. 

\begin{proof}[Proof of \Cref{thm_main}]
For $d|m$ direction we repeat $|m|/d$ times the algorithm of \cite{dong_controlled} (\Cref{section_cUd}), which exactly achieves $(c^d_\phi, \{id\})$ using $d$ queries. We switch the $id$ to $inv$ queries if $m<0$.

For the $d\nmid m$ direction, assume an algorithm $(\mathcal{A}, \{id, inv\})$ $\epsilon$-approximately achieves $(c^m_\phi, \{id, inv\})$ for some $\epsilon<1/2$. As mentioned, $\mathcal{A}$ is a continuous, $0$-homogeneous function of $U\in U(d)$. Consider the input state $\rho\in\mathcal{D}(\mathcal{H})$, $\rho=\ket{+}\!\!\bra{+}_c\otimes\ket{0}\!\!\bra{0}$, where the first register is the control qubit, and the second the target qudit. Using $\rho$ define $f: U(d)\to \mathbb{C}$ as
\begin{align*}
    f(U) :&=  \operatorname{tr}\left[\frac{\mathcal{A}(U)(\rho)}{\operatorname{tr}\left[\mathcal{A}(U)(\rho)\right]}\left(\ket{1}\!\!\bra{0}_c\otimes U^m\right)\right]\text{,}
\end{align*}
which is continuous, $m$-homogeneous. Observing that 
\begin{equation}
    \operatorname{tr}\left[c^m_\phi(U)\rho \,c^m_\phi(U)^\dagger\left(\ket{1}\!\!\bra{0}_c\otimes U^m\right)\right]=\frac{e^{-i\phi(U)}}{2}\text{,}
\end{equation}
and that multiplication by $\ket{1}\!\!\bra{0}_c\otimes U^m$ and taking trace cannot increase trace norm, we get 
\begin{align*}
    \left|f(U)-\frac{e^{-i\phi(U)}}{2}\right|&\leq\left|\left|
    \frac{\mathcal{A}(U)(\rho)}{\operatorname{tr}\left[\mathcal{A}(U)(\rho)\right]}
    -c^m_\phi(U)\rho \,c^m_\phi(U)^\dagger\right|\right|_{\operatorname{tr}}\text{,}
\end{align*}

\noindent which is less than $\epsilon$ by the assumption. Thus
\begin{equation*}
    |f(U)|\geq \frac{1}{2}-\epsilon >0 \text{,}
\end{equation*}
for all $U\in U(d)$, because $\epsilon <1/2$ by the assumption. This lets us define the continuous function $\widehat{f}(U):=f(U)/|f(U)|$. Recall that homogeneity of functions on $U(d)$ is defined with respect to scalars of unit norm only, $\lambda\in S^1$. Thus, $|f(U)|$ is $0$-homogeneous, and the continuous function $\widehat{f}:U(d)\to S^1$ is $m$-homogeneous. By \Cref{thm_top} $d|m$.
\end{proof}

\begin{proof}[Proof of \Cref{thm_top}]
    This elementary proof (for a shorter proof see \Cref{section_top_short}) studies the structure of the space $U(d)$, finding paths on this space that are continuously deformable, or \emph{homotopic}, to each other. Denoting a path by $U(t)$, $t\in [0, 1]$, we split $[0, 1]$ to $d$ identical intervals labeled by $j\in\{0,1\dots, d-1\}$ and use the parameter $\Delta\in[0, 1]$ to move inside each interval. Consider the paths $U(t)=e^{i2\pi t}I$,
\begin{equation*}
    U'\left(\frac{j+\Delta}{d}\right)=\bordermatrix{
    & {\scriptstyle{0}} & {\scriptstyle{\dots}} & {\scriptstyle{j}} & {\scriptstyle{\dots}} & \quad{\scriptstyle{d-1}}\cr
    {^{0}}&\ddots && &  &  \cr
    {\scriptstyle{\vdots}}& &  1  &  &  &  \cr
    {\scriptstyle{j}} & &  & e^{i2\pi\Delta} && &\cr
    {\scriptstyle{\vdots}} &&  &  & 	1  &  \cr
    {\scriptstyle{d-1}}& &  &  &  & \ddots \cr
    }
\end{equation*} 
and 
\begin{equation*}
    U''\left(\frac{j+\Delta}{d}\right)=\begin{pmatrix}
    e^{i2\pi\Delta}  &  &  & \cr
    & 1 && \cr
    &  &  \ddots  &  \cr
    &  &  & 1 
    \end{pmatrix}\text{,}
\end{equation*}
where all the off-diagonal matrix elements are zero and the diagonal elements change: $U(t)$ loops on all $d$ of them simultaneously, $U'(t)$ loops on all $d$ of them in sequence, and $U''(t)$ loops only on the zeroth one.

First, we show that $U$ is homotopic to $U'$. We can write $U(t)=\exp(i2\pi \operatorname{diag}[t,t\dots t])$ and $U'(\tfrac{j+\Delta}{d})=\exp(i2\pi \operatorname{diag}[1 \dots 1, \Delta, 0\dots 0])$, where $\exp$ is matrix exponentiation and $\operatorname{diag}v$ is the square matrix with the vector $v$ on the diagonal and zeros elsewhere. For any vector in the $d$-dimensional hypercube, $v\in[0,1]^d$, the continuous function $v\mapsto \exp(i2\pi\operatorname{diag} v)$ outputs a $d$-dimensional unitary. If fed a hypercube path $v(t)$ from $[0,\dots, 0]$ to $[1,\dots, 1]$, the function outputs a unitary path starting and ending at the identity: the unitary path $U(t)$, if the hypercube path is the straight line; the unitary path $U'(t)$, if the hypercube path travels along the edges of the hypercube (Fig. \ref{fig_proof_hypercube}). Since the two hypercube paths are homotopic, so are $U$ and $U'$.

\begin{figure}
         \centering
         \includegraphics[]{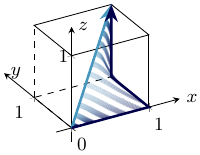}\vspace{-10pt}
         \caption{Sequential and simultaneous increase of each coordinate are homotopic inside the cube $[0,1]^3$.}
         \label{fig_proof_hypercube}
\end{figure}

To show that $U'$ and $U''$ are homotopic, note that $U'(\frac{j}{d})=U''(\frac{j}{d})$ for all $j$. Therefore, showing their homotopy on each of the $d$ subintervals suffices. Indeed, there exists a continuous unitary path $R_j(s)$, $s\in[0,1]$ such that conjugating by $R_j(s)$ continuously deforms $U'$ into $U''$ on the $j$th subinterval; at $s=0$ $U'$ stays unchanged, at $s=1$ it becomes $U''$ -- the role of the path $R_j(s)$ being analogous to the thick lines in Fig. \ref{fig_proof_hypercube}.
\vspace{-1em}
\begin{align*}
R_j(s)&=
\bordermatrix{
    &&&&&&& \cr
    {\scriptstyle{0}} & \cos\left(\frac{\pi}{2}s\right) &  &  &  & i\sin\left(\frac{\pi}{2}s\right)&&\cr
    {\scriptstyle{1}} & & 1 &&&&&\cr
    {\scriptstyle{\vdots}} & &  & \ddots &&&&\cr
    & &  &  & 1 &&&\cr
    {\scriptstyle{j}} &i\sin\left(\frac{\pi}{2}s\right) &  &  &  & \cos\left(\frac{\pi}{2}s\right) &&\cr
    &&  &  &  &  & 1 &\cr
    {^{\vdots}}& &  &  &  &  &  & \ddots\cr
    }
\end{align*}

\noindent so, indeed, the map $U'(\tfrac{j+\Delta}{d})\mapsto R_j(s)U'(\tfrac{j+\Delta}{d})R_j(s)^\dagger$ at $s=0$ does nothing and at $s=1$ swaps the $0$th and the $j$th basis state. This fixes a homotopy between $U'$ and $U''$. Combining with the first homotopy, we conclude that $U$ and $U''$ are homotopic.

Assuming a continuous $m$-homogeneous function $f:U(d)\rightarrow S^{1}$, next we show that $d$ must divide $m$. Since $f$ is $m$-homogeneous, we have $f(U(t))=e^{i m2\pi t}f(I)$, which corresponds to $m$ loops on the circle $S^1$ starting and ending at $f(I)$. $U''(t)$ is a periodic function with $d$ periods inside $[0,1]$, and so is the function $f(U''(t))$. Therefore, if $f(U''(t))$ makes $k\in\mathbb{Z}$ loops on the interval $[0,\tfrac{1}{d}]$, then it makes $kd$ loops on $S^1$ on the entire $[0,1]$ interval. Since $f$ is continuous, $f(U(t))$ must be homotopic to $f(U''(t))$. But on the circle $S^1$ two paths can be homotopic if and only if they make the same number of loops (this is captured by the fundamental group of $S^1$). We must have $m=kd$, which proves \Cref{thm_top}.
\end{proof}

\subsection{Process tomography for the if clause?}

\Cref{thm_main} seems to contradict process tomography. After we allow approximations and arbitrarily many queries, process tomography seems to offer a solution to the if clause. $\exp(n)$ queries suffice to estimate any oracle's \emph{superoperator} with exponentially small error \cite{surawy2022projected}. To build the if clause, our process tomography should output a matrix estimate $X\in U(d)$, instead of the superoperator estimate $\rho \mapsto X\rho X^\dagger$, because controlled-$X$,
\begin{equation}
    \rho\mapsto \left(\ket{0}\!\!\bra{0}\otimes I+\ket{1}\!\!\bra{1}\otimes X\right) \rho \left(\ket{0}\!\!\bra{0}\otimes I + \ket{1}\!\!\bra{1}\otimes X\right)^\dagger\text{,}
\end{equation}

\noindent contains cross-terms. Thus, in this section we study process tomography (of a unitary oracle) that outputs a \emph{matrix} estimate $X$.

If such tomography is possible, its matrix output $X$ approximately describes the oracle $U\in U(d)$ only up to the global phase, because distinguishing the global phase is unphysical. Thus, to quantify the error, we compare $\kappa(X)$ to $\kappa(U)$, where the function $\kappa: U(d)\to U(d)$ fixes one canonical form for all unitaries equal up to the global phase. Formally, $\kappa(e^{i\alpha}U)=\kappa(U)$ for all $\alpha\in[0, 2\pi)$ ($\kappa$ is $0$-homogeneous) and $\kappa(U)=e^{i\phi(U)}U$ for some real function $\phi$. An example is a function that makes the first nonzero matrix element of its input real positive.

While any practical process tomography procedure includes classical postprocessing of measurement outcomes, in principle it is convertible to a fully quantum algorithm with one measurement at the end that yields the full estimate -- in our case a full matrix $X$: all its matrix elements written down to some $N$-dependent precision. Being a measurement outcome, this matrix estimate $X$ is a random variable that takes values in the discrete space $\Omega_N\subset U(d)$.  

\begin{mydef}[Process tomography yielding a unitary matrix]\label{def_tom}
Process tomography of a unitary oracle is an algorithm making $N$ queries to $U\in U(d)$ and then measuring, obtaining the outcome $X\in\Omega_N\subset U(d)$ with probability $p_{N,U}(X)$. The output distribution $p_{N,U}$ is such that for all $U\in U(d)$
\begin{equation*}
    \operatorname{Pr}_{X\sim p_{N,U}(\cdot)}\left[\left|\left|\kappa(X)-\kappa(U)\right|\right|_{\operatorname{op}}\leq \epsilon_N\right]\geq 1-\delta_N\text{,}
\end{equation*}
where $\lim_{N\to\infty}\epsilon_N=0$ and $\lim_{N\to\infty}\delta_N=0$.
\end{mydef}

\noindent The \emph{operator} norm $\left|\left|\cdot\right|\right|_{\operatorname{op}}$ distinguishes \Cref{def_tom} from existing process tomography algorithms that ensure a vanishing distance of \emph{superoperators} \cite{surawy2022projected}. Can we relate the two distances? Vanishing operator norm implies vanishing diamond distance of the corresponding superoperators \cite[Lemma 12.6]{mixed_states_aharonov98}. Conversely, vanishing diamond distance $\left|\left|X\cdot X^\dagger-U\cdot U^\dagger \right|\right|_\diamondsuit$ implies vanishing 
\begin{equation}
    \left|\left|X-\bra{v}U^\dagger X\ket{v}U\right|\right|_{\operatorname{op}}=\left|\left|\kappa(X)-\bra{v}U^\dagger \kappa(X)\ket{v}U\right|\right|_{\operatorname{op}}\label{eq_X-dep_phase}
\end{equation}
for any unit vector $\ket{v}$ \cite[Theorem 1.3]{technical}. However, \Cref{def_tom} requires the $U$-term to be independent of $X$. This makes all the difference.

\begin{mytheorem}\label{result_state_tom}
    Process tomography of \Cref{def_tom} is impossible.
\end{mytheorem}

\begin{proof}
The triangle inequality gives
\begin{eqnarray}
   \left|\left|\kappa(U)-\kappa(V)\right|\right|_{\operatorname{op}}
    &\leq & \sigma_{{N,U}} + \!\sum_{X\in\Omega_N}\!\left|p_{N,U}(X)-p_{N,V}(X)\right|\nonumber\\
    && + \sigma_{{N,V}}\text{,}\label{eq:tom}
\end{eqnarray}
where $\sigma_{N,U}:=\sum_{X\in\Omega_N}p_{N,U}(X)\left|\left|
    \kappa(U)-\kappa(X)\right|\right|_{\operatorname{op}}$ is a sum which can be split to two sums -- over “good” and “bad” $X$-values -- and then by \Cref{def_tom} upper bounded by $\epsilon_N+2\delta_N$. The middle term of \eqref{eq:tom} is the total variational distance of two probability distributions. Equivalently, it is the trace distance of two classical density states $\rho_{N,U}=\operatorname{diag}(p_{N,U})$ and $\rho_{N,V}=\operatorname{diag}(p_{N,V})$, the outputs of the process tomography algorithm making $N$ queries to $U$ or $V$ respectively. By the sub-additivity of diamond distance, the trace distance of the outputs is upper bounded $||\rho_{N,U}-\rho_{N,V}||_{\operatorname{tr}}\leq N||U\cdot U^\dagger - V\cdot V^\dagger||_\diamondsuit$. Using the inequality $||U\cdot U^\dagger - V\cdot V^\dagger||_\diamondsuit\leq 2||U-V||_{\operatorname{op}}$ \cite[Lemma 12.6]{mixed_states_aharonov98} we get 
\begin{equation}
   \left|\left|\kappa(U)-\kappa(V)\right|\right|_{\operatorname{op}}\leq 
    2N\left|\left|U - V\right|\right|_{\operatorname{op}} + 2\epsilon_N + 4\delta_N\text{.}\label{eq_cont}
\end{equation}
Thus $\kappa$ is continuous. 
But if $\kappa$ is continuous (it is $0$-homogeneous by definition), then defining $f(U):=\bra{0}\kappa(U)U^\dagger\ket{0}=e^{i\phi(U)}$ gives us $f:U(d)\to S^1$ continuous and $-1$-homogeneous, contradicting \Cref{thm_top}.
\end{proof}

Building on top of known process tomography algorithms cannot produce matrix estimates in the sense of \Cref{def_tom} and the proof reveals the reason: The arbitrary closeness of the \emph{matrices} $\kappa(X)$, $\kappa(U)$ cannot be guaranteed because any $\kappa$ must be discontinuous (see also \Cref{lem:cover1} in \Cref{section_cover_spaces}). Consider the $\kappa$ example that makes the first nonzero matrix element real positive. This $\kappa$ is discontinuous. At the discontinuity is, for example, the Pauli-X operator $U=\kappa(U)=\sigma_x$, because its first matrix element is zero, $\langle 0 |\sigma_x | 0\rangle =0$. We can choose an estimate $X$ with small but negative $\langle 0 |X | 0\rangle$, so that $\kappa(X)\approx -\sigma_x$. This maximizes the operator norm difference $\left|\left|\kappa(X)-\sigma_x\right|\right|_{\operatorname{op}}\approx 2$. Thus, near the discontinuity, $\smash{\left|\left|X\cdot X^\dagger-U\cdot U^\dagger \right|\right|_\diamondsuit}$ vanishes, but $\left|\left|\kappa(X)-\kappa(U)\right|\right|_{\operatorname{op}}$ is large -- the discontinuity in $\kappa$ amplifies the error arbitrarily.

We suggest modifying \Cref{def_tom} to use multiple $\kappa_j$, each continuous on some subset of $U(d)$ and applied only if the sampled $X$ lies deep inside this subset. For example, the following functions, indexed by $j\in\{0,1\dots, d-1\}$, split $U(d)$ to $d$ such (overlapping) subsets: 
\begin{equation}
    \kappa_j(U):=  
    \begin{cases}
    \frac{\overline{{\langle j |U|0\rangle}}}{|\langle j |U|0\rangle|}U & \text{if } \langle j |U|0\rangle\neq 0 \\
    \kappa_{j+1 \operatorname{mod}d}(U) & \text{otherwise}\label{eq:kappa_ex}
    \end{cases}
\end{equation}

\noindent Having sampled $X$ choose $\kappa_r$, where $r=r(X)=\min\{j: |\langle j |X|0\rangle|\geq 1/\sqrt{d}\}$. 
The chosen function is continuous around the sampled $X$, circumventing the impossibility. Thus given some set $\{\kappa_r\}_{r}$ and a rule $r=r(X)$ we define

\begin{mydef}[Process tomography yielding a unitary matrix -- revised]\label{def_tom2}
Process tomography of a unitary oracle is an algorithm making $N$ queries to $U\in U(d)$ and then measuring, obtaining the outcome $X\in\Omega_N\subset U(d)$ with probability $p_{N,U}(X)$. The output distribution $p_{N,U}$ is such that for all $U\in U(d)$
\begin{equation*}
    \operatorname{Pr}_{X\sim p_{N,U}(\cdot)}\left[\left|\left|\kappa_r(X)-\kappa_r(U)\right|\right|_{\operatorname{op}}\leq \epsilon_N\right]\geq 1-\delta_N\text{,}
\end{equation*}
where $\lim_{N\to\infty}\epsilon_N=0$, $\lim_{N\to\infty}\delta_N=0$ and $r=r(X)$.
\end{mydef}

\begin{table}
\caption{If-clause achievability by algorithms that use a different type of measurement}\label{table_det-random}
\begin{ruledtabular}
\begin{tabular}{ l l l } 
  measurement with &  if clause  &  random if clause \\
\hline
 one success outcome & \ding{56} & \ding{56} (\Cref{section_relaxed}) \\
 many success outcomes  & \ding{56} (\Cref{section_relaxed}) & \ding{52}\! (\Cref{section_def_tom2_exist},\ref{section_relaxed})
\end{tabular}
\end{ruledtabular}
\end{table}

\Cref{section_def_tom2_exist} argues that any standard process tomography \cite{surawy2022projected} combined with \eqref{eq:kappa_ex} satisfies this definition. The dependence of $\kappa_r$ on the sampled $X$ yields randomness beyond the estimation error: There are oracles $U$ for which the sampling yields either an approximation of $\kappa_0(U)$ or an approximation of $\kappa_1(U)$, both with a non-negligible probability. (Consider the Fourier transform oracle $\langle j | U |k\rangle = e^{i\frac{2\pi}{d}jk}/\sqrt{d}$ and $\kappa_r$, $r(X)$ of example \eqref{eq:kappa_ex}.) Having obtained a classical estimate of $\kappa_r(U)$, we build a circuit close to $\ket{0}\!\!\bra{0}\otimes I + \ket{1}\!\!\bra{1}\otimes \kappa_r(U)=c_{\phi_r}(U)$, but the value $r$ changes the circuit's operator (and the superoperator) beyond the estimation error. The index $r$ is known and random; calculated from the measurement outcome obtained along implementing $c_{\phi_r}(U)$, instead of the if clause, process tomography yields the \emph{random if clause} (see \Cref{section_relaxed}). A many-outcome measurement is necessary for randomness. Thus, on the random if clause, quantum circuits using a many-outcome measurement are more powerful than those using a binary success/fail measurement (\Cref{table_det-random}). 

We can make the tomography-based strategy fully quantum, deferring the $r$ measurement to an additional ancilla register. Before the measurement, this implements a superoperator close to
\begin{equation}
    \rho\mapsto \sum_r p_{r \,U}\left[c_{\phi_r}(U)\otimes\ket{r}\right]\rho \left[c_{\phi_r}(U)^\dagger\otimes \bra{r}\right]\text{,}\label{eq:main_entangled}
\end{equation}
corresponding to the \emph{entangled if clause} (see \Cref{section_relaxed}).

After the deferring, a $d$-outcome measurement suffices to achieve the random if clause, because $r$ in 
$c_{\phi_r}(U)=\ket{0}\!\!\bra{0}\otimes I + \ket{1}\!\!\bra{1}\otimes e^{i\phi_r(U)}U$ takes $d$ values. This measurement is less complex than the full process tomography \footnote{With the full process tomography, one can also implement $c_{\phi_X}(U)$ with $X\in\Omega_N$, because the outcomes $X$ yield $\ket{0}\!\!\bra{0}\otimes I + \ket{1}\!\!\bra{1}\otimes X$ which is close to $c_{\phi_X}(U)=\ket{0}\!\!\bra{0}\otimes I + \ket{1}\!\!\bra{1}\otimes e^{i\phi_X(U)}U$ via expression \eqref{eq_X-dep_phase} and by choosing $\ket{v}$ to be an eigenvector of $U^\dagger X$ so that $\bra{v}U^\dagger X\ket{v}=:e^{i\phi_X(U)}$.}. Since the $|\Omega_N|$-outcome tomography measurement turns out to be unnecessary, one may ask whether the $d$-outcome measurement is optimal or can be simplified further. Even more pressing is the question of the query complexity of the random if clause. The brute-force tomography approach is likely far from optimal.

\section{Applications to other tasks}\label{section_tasks}

Here we use our method to study other tasks, namely the neutralization, fractional power, transpose, and inverse task. For each task we first review some preceding work. 

\subsection{Review of the $c_\phi^d$ algorithm via the neutralization}\label{section_cUd}
In this section we review the unitary algorithm of Dong, Nakayama, Soeda, and Murao \cite{dong_controlled} that exactly implements $c_\phi^d(U)$. Suppose we try to implement $c_\phi^m(U)$ taking the following approach: 
(I) use the controlled-swap gates as in \Cref{fig_kitaev_vacuum}; if the control is $\ket{1}$ the oracle $U$ is queried in sequence $m$ times, if the control is $\ket{0}$ the $m$ queries are 'moved' to an ancilla register.
(II) If the control is $\ket{0}$ apply to the ancilla a subroutine that makes the 'moved' queries have no effect other than a multiplication by a scalar.
In the words of ref. \cite{dong_controlled}, the subroutine of step II \emph{neutralizes} the 'moved' queries. Fitting the neutralization into our framework we define:

\begin{mydef}[$\epsilon$-approximately neutralizing a query sequence]\label{def_neutr}
Given a query sequence $\mathbf{s}$, a postselection oracle algorithm $A_\varnothing:U(d)\to L( \mathcal{H})$ {\it $\epsilon$-approximately neutralizes} $\mathbf{s}$, if its query sequence is $\mathbf{s}$ and still it leaves an all-zero input (almost) untouched:
\begin{equation}
    \left|\left|\frac{A_\varnothing(U)\ket{\mathbf{0}}\!\!\bra{\mathbf{0}}A_\varnothing(U)^\dagger}{\operatorname{tr}\left[A_\varnothing(U)\ket{\mathbf{0}}\!\!\bra{\mathbf{0}}A_\varnothing(U)^\dagger\right]}-\ket{\mathbf{0}}\!\!\bra{\mathbf{0}}\right|\right|_{\operatorname{tr}}\leq \epsilon\text{.}\label{eq_eps-neutralize}
\end{equation}
\end{mydef}

\noindent \Cref{def_neutr} uses the algorithm's operator instead of the superoperator in order to forbid the (partial) trace; applying the oracles to an ancilla that is later traced out is not considered a neutralization.

\begin{figure*}
     \subfloat[]{
     \includegraphics[valign=c]{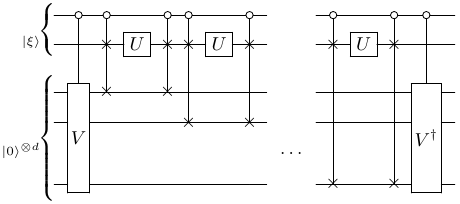}
     \vphantom{\includegraphics[valign=c]{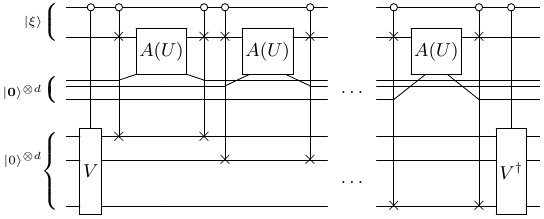}}
\label{fig_cUd}
     }
     \subfloat[]{
     \includegraphics[valign=c]{fig6b_cUapprox-cropped.pdf}
\label{fig_cU_approx}
     }
     \caption{(a) The unitary algorithm due to \protect \cite{dong_controlled} makes $d$ queries to $U$ and implements $c_\phi^d(U)$ to its first two registers. (b) The queries to $U$ are replaced by subroutines $A(U)$ that possibly act on additional ancillae. $A(U)$ is assumed to approximately apply $t(U)=U^\frac{1}{d}$ to its non-ancilla Hilbert space.
}\label{fig_cUd_both}
\end{figure*}

Observing that the totally antisymmetric state is an eigenstate of $U^{\otimes d}$, ref.  \cite{dong_controlled} presented an algorithm \emph{exactly neutralizing} ($\epsilon=0$) the sequence $\mathbf{s}=(id)^d$ of \emph{parallel} queries. This postselection oracle algorithm, $\smash{A^{\parallel d}_\varnothing(U)}$, is unitary. Ref. \cite{dong_controlled} used $\smash{A^{\parallel d}_\varnothing(U)}$ to build the algorithm of Fig. \ref{fig_cUd}, the (unitary) postselection oracle algorithm that exactly achieves $(c_\phi^d, \{ id\})$:

\begin{myconstr}[due to  \cite{dong_controlled}]
For the correctly chosen fixed unitary $V$, the unitary algorithm of Fig. \ref{fig_cUd} makes $d$ queries to $U$ and implements the operator $A(U)$ such that $A(U)\left( I_{\mathcal{H}} \otimes\ket{\mathbf{0}}\!\!\bra{\mathbf{0}}_ \mathcal{K}\right)$ equals
\begin{equation}
 \left(\ket{0}\!\!\bra{0}\otimes I + \det(U)^{-1}\ket{1}\!\!\bra{1}\otimes U^{d}\right)\otimes\det(U)\ket{\mathbf{0}}\!\!\bra{\mathbf{0}}_ \mathcal{K}\text{.}\label{eq_cUd}
\end{equation}
\end{myconstr}

\begin{proof}
First note that we can write the determinant as
\begin{align}
    \det\left(U\right)&=\sum_{{\pi} \in S_{d}}\operatorname{sgn}\left({\pi}\right)\prod_{i=1}^{d}\bra{i}U\ket{{\pi}\!\left(i\right)}\nonumber\\
    &=\frac{1}{d!}\sum_{{\pi},\tau\in S_{d}}\operatorname{sgn}\left(\tau\right)\operatorname{sgn}\left({\pi}\right)\prod_{i=1}^{d}\bra{\tau\!\left(i\right)}U\ket{{\pi}\!\left(i\right)}\label{eq_det}\\
    &=\bra{\chi_d}U^{\otimes d}\ket{\chi_d}\nonumber\text{,}
\end{align}
 where 
 $\smash{\ket{\chi_d}=\frac{1}{\sqrt{d!}}\sum_{{\pi}\in S_{d}}\operatorname{sgn}\left({\pi}\right)
 \ket{{\pi}\!\left(1\right)}\otimes\dots\otimes\ket{{\pi}\!\left(d\right)}\in \mathcal{K}}$ is the totally antisymmetric state. For a proof of \cref{eq_det} see \Cref{section_minors}. Note that we have indexed the basis vectors starting with $i=1$. Switching to the usual labeling that starts with zero, observe that $\ket{\chi_d}$ is normalized and therefore there exists a unitary $V\in L( \mathcal{K})$ such that 
 $\smash{V\ket{\mathbf{0}}=\ket{\chi_d}}$. Define the unitary oracle algorithm 
 $\smash{A^{\parallel d}_\varnothing(U):=V^\dagger U^{\otimes d} V}$. 
 Since 
 $\smash{|\bra{\chi_d}U^{\otimes d}\ket{\chi_d}|=1}$, the unitary $\smash{A^{\parallel d}_\varnothing(U)}$ on input $\ket{\mathbf{0}}$ outputs $\ket{\mathbf{0}}$ with probability $1$ for all $U\in U(d)$. Therefore $\smash{A^{\parallel d}_\varnothing}$ exactly neutralizes its queries, $\smash{A^{\parallel d}_\varnothing(U)\ket{\mathbf{0}}\!\!\bra{\mathbf{0}}=\det(U)\ket{\mathbf{0}}\!\!\bra{\mathbf{0}}}$. Observe that the algorithm of Fig. \ref{fig_cUd} implements $\smash{\ket{0}\!\!\bra{0}\otimes  I \otimes A^{\parallel d}_\varnothing(U)_ \mathcal{K}+ \ket{1}\!\!\bra{1}\otimes U^d\otimes I_ \mathcal{K}}$, from which \eqref{eq_cUd} follows.
\end{proof}

\subsection{Neutralizable query sequences}\label{section_neutr}

Complementary to the neutralizing algorithm $\smash{A^{\parallel d}_\varnothing(U)}$ above, here we prove the impossibility of neutralizing the remaining numbers of parallel, as well as sequential $id$ queries. At first glance, it seems that one could always add a query and then make it have no effect in the algorithm. This is less obvious when tracing out subsystems is disallowed, as in the neutralization (\Cref{def_neutr}). Indeed, for most query numbers, the neutralization is impossible. Conceptually, this no-neutralization of oracle queries is analogous to the no-deleting of states \cite{no_deleting}.

\begin{mytheorem}\label{thm_neutr}
Let $d\in \mathbb{N}$ be the dimension of the oracle $U$, let $\boldsymbol{s}=(id)^m$ and $\epsilon\in[0,1)$. 
If there exists a post\-selection oracle algorithm $\epsilon$-approximately neutralizing $\boldsymbol{s}$, then $d|m$.
\end{mytheorem}

\begin{proof}\label{proof_neutr2}
Suppose $(A_\varnothing, \{id\})$ is a postselection oracle algorithm $\epsilon$-approximately neutralizing $\boldsymbol{s}$. By \cref{eq_general_alg}, $A_\varnothing:U(d)\to L( \mathcal{H})$ is continuous and $m$-homogeneous and so is
\begin{equation*}
    f_{A_\varnothing}(U):=\frac{\bra{\mathbf{0}}A_\varnothing(U)\ket{\mathbf{0}}}{\sqrt{\operatorname{tr}\left[A_\varnothing(U)\ket{\mathbf{0}}\!\!\bra{\mathbf{0}}A_\varnothing(U)^\dagger\right]}}\text{,}
\end{equation*}
Since taking $X\mapsto \bra{\mathbf{0}}X\ket{\mathbf{0}}$ can only contract trace norm, inequality \eqref{eq_eps-neutralize} implies that
\begin{equation*}
    \left|\left|f_{A_\varnothing}(U)\right|^2-1\right|\leq \epsilon <1
\end{equation*}
so  $\left|f_{A_\varnothing}(U)\right|\neq 0$ holds for all $U\in U(d)$ and we can define $\widehat{f}_{A_\varnothing}(U):=f_{A_\varnothing}(U)/\left|f_{A_\varnothing}(U)\right|$ which maps from $U(d)$ to the circle and is continuous $m$-homogeneous. Then by \Cref{thm_top}, $m$ is a multiple of $d$.
\end{proof}

\noindent \Cref{thm_neutr} does not distinguish between parallel and sequential queries -- it applies to both. More generally, all of the above applies also to process-matrix algorithms (a model already without the parallel/sequential distinction).

\subsection{Fractional power impossibility}\label{section_worst-case} 
Here we are interested in the $1/d$th power task $(U\mapsto U^{\frac{1}{d}}, \{id, inv\})$. It is listed in \Cref{table_tasks2}. For this task, Sheridan, Maslov and Mosca \cite{sheridan_maslov_mosca} found an average-case algorithm -- an algorithm that works for most, but not all $U\in U(d)$. Given any fraction $q\in\mathbb{Q}$, the algorithm makes $id$, $inv$ queries (after the adjustment in Appendix A of \cite{sheridan_maslov_mosca} that removes the controlled-$U$ queries) and implements an operator that for most $U\in U(d)$ is close to $U^q\otimes U^{-q}\otimes  I_{\mathcal{K}_2}$. This can be arbitrarily close at the cost of increasing the dimension of the ancilla Hilbert space $\mathcal{K}_2$. In \cite{sheridan_maslov_mosca} $U\mapsto U^q$ corresponds to a specific fractional-power function. Many exist, because each unitary has many possible roots. For example, for any $j,k\in \{0, 1,\dots m-1\}$ a possible $m$th root of 
\begin{equation*}
    \begin{pmatrix}
        1 & 0\\
        0 & -1
    \end{pmatrix}\text{ is }\begin{pmatrix}
        e^{i\frac{2\pi j}{m}} & 0\\
        0 & e^{i\frac{\pi + 2\pi k}{m}}
    \end{pmatrix}\text{.}
\end{equation*}

Here we show that, unlike in the average case, in the worst case we cannot implement $U^{\frac{1}{d}}$ to within an arbitrarily small error $\epsilon$. For the specific $1/d$th power function chosen in \cite{sheridan_maslov_mosca}, the impossibility was already proven in \cite[Lemma 1]{sheridan_maslov_mosca}. However, the following corollary of \Cref{thm_main} holds for \emph{any} $1/d$th power function.

\begin{mycor}\label{thm_root} Let $t:U(d)\to U(d)$ be any function such that $t(U)^d=U$. No postselection oracle algorithm or process-matrix algorithm can $\epsilon$-approximately achieve the $1/d$th power task $(t, \{id, inv\})$ for $\epsilon< 1/2d$, no matter how many queries it makes.
\end{mycor}

The impossibility of errors smaller than $1/\exp(n)$ for $n$-qubit unitaries is discouraging, because the “inverse exponential scaling” here is qualitatively different from how errors scale in useful algorithms. In a useful algorithm the error $\epsilon$ scales inverse exponentially  or  even inverse polynomially with respect to the \emph{complexity} of the algorithm. In other words, the complexity scales as $\log_2(1/\epsilon)$ or $\operatorname{poly}(1/\epsilon)$ \cite{gilyen2019quantum,lloyd1996universal,bacon2006efficient,van2002efficient}. Here, instead, we found $\epsilon$-values for which arbitrarily complex algorithms fail.

Another way to understand the strength of the result is to consider only two-dimensional unitary oracles, which intuitively should be the simplest to work with. According to the result, the square root of a two dimensional unitary is impossible to implement within $\epsilon<1/4$ by quantum circuits of any complexity.

The intuition for a complexity-independent lower bound on $\epsilon$ is simple: we are trying to approximate a discontinuous function $t$ by some {\it continuous} algorithm.

\begin{proof}
Suppose $(A, \{id, inv\})$ $\epsilon$-approximately achieves the $1/d$th power task $(t, \{id, inv\})$. Use $A(U)$ instead of each query to $U$ in the $(c_\phi^d, \{id, inv\})$ algorithm of ref.  \cite{dong_controlled} as in Fig. \ref{fig_cUd_both}. The resulting algorithm in the exact ($\epsilon=0$) case $A(U)\left(I\otimes \ket{\mathbf{0}}\!\!\bra{\mathbf{0}}\right)=t(U)\otimes \ket{g(U)}\!\!\bra{\mathbf{0}}$ implements
\begin{eqnarray*}
    &\left(\ket{0}\!\!\bra{0}\otimes I + \det(t(U))^{-1}\ket{1}\!\!\bra{1}\otimes U\right)\!\!\!&\otimes \left(\ket{g(U)}\!\!\bra{\mathbf{0}}\right)^{\otimes d}\\
    &&\otimes \det(t(U))\ket{\mathbf{0}}\!\!\bra{\mathbf{0}}\text{,}
\end{eqnarray*}

\noindent which corresponds to the if clause. For $A$ with error $\epsilon$, the composed algorithm in Fig. \ref{fig_cU_approx} has error $d\epsilon$, because it uses the $A(U)$ subroutine $d$ times and because postselection diamond distance satisfies weak subadditivity (Theorem 2.1 in \cite{technical}). By \Cref{thm_main} we must have $d\epsilon\geq 1/2$ implying that $\epsilon\geq 1/2d$.
\end{proof}

\subsection{Review of the complex conjugation algorithm}\label{section_compl_conj}
Complex conjugation is used as a subroutine in the inverse algorithm discussed in the subsequent sections. Here we review the complex conjugation algorithm of Miyazaki, Soeda and Murao  \cite{compl_conj}, giving an alternative proof of its correctness.

\begin{myconstr}[due to  \cite{compl_conj}]\label{thm_conj}
For a correctly chosen fixed unitary $V_{conj}$, the unitary algorithm in the dashed box in Fig. \ref{fig_Udagger} implements $A_\text{conj}(U)$ such that
\begin{equation}\label{eq_complconj}
   A_\text{conj}(U)\left( I\otimes\ket{0}\!\!\bra{0}^{\otimes d-2}\right)= \overline{U}\otimes\det(U)\ket{0}\!\!\bra{0}^{\otimes d-2}
\end{equation} 
while making $d-1$ queries to $U$.
\end{myconstr}

Therefore, the (unitary) postselection oracle algorithm exactly achieves the complex conjugation task $(U\mapsto \overline{U}, \{id\})$.

\begin{figure}
\centering
\includegraphics[]{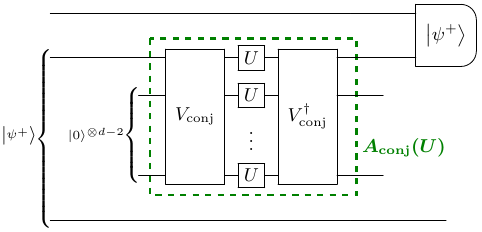}
\caption{After postselecting the first two registers on outcome $\ket{\psi^+}$, this algorithm of \protect{\cite{quintino_inverse, quintino_transforming}} implements $U^\dagger$. Its subroutine in the dashed box, the unitary algorithm of \protect{\cite{compl_conj}}, implements the complex conjugate $\overline{U}$ from  $d-1$ parallel queries to $U$.
}\label{fig_Udagger}
\end{figure}

\begin{proof}
Cramer's rule for matrix inversion reads $U^{-1}=\det(U)^{-1}C(U)^T$, where $C(U)$ is the cofactor matrix of $U$. We define it presently, but first note that from Cramer's rule $C(U)=\det(U) \overline{U}$. It remains to show that the algorithm implements $C(U)$ to its first register corresponding to a $d$-dimesional Hilbert space $\mathcal{H}$. Cofactor matrix of $U$ is the matrix of minors of $U$, $C(U)_{i,j}=(-1)^{i+j}\det(U_{\centernot{i},\centernot{j}})$ or:
\vspace{-10pt}
\begin{equation}
C(U) = \sum_{i, j = 1}^d(-1)^{i+j}\det(U_{\centernot{i},\centernot{j}})\ket{i}\!\!\bra{j}\text{,}\label{eq_cofactor}
\end{equation}
\noindent where the labeling of the computational basis vectors starts with $i=1$, and where $\det(U_{\centernot{i},\centernot{j}})$ is the $(i,j)$th minor, defined as the determinant of the matrix $U$ with the $i$th row and $j$th column deleted. Formally, we write
\begin{equation}
\det(U_{\centernot{i},\centernot{j}})\!=\frac{(-1)^{i+j}}{(d-1)!}\!\!\sum_{\substack{\tau,{\pi}\in S_d\\\tau(1)=i\\{\pi}(1)=j}}\!\!\!\operatorname{sgn}(\tau)\operatorname{sgn}({\pi})\prod_{k=2}^{d}\!\!\bra{\tau(k)}U\ket{{\pi}(k)}\text{.}\label{eq_minors}
\end{equation}
See \Cref{section_minors} for the proof. Substituting \eqref{eq_minors} into \eqref{eq_cofactor} we get $C(U)= E^\dagger U^{\otimes d-1} E$ with $E=\sum_{j=1}^d\ket{v_j}\!\bra{j}$ and
\begin{equation*}
\ket{v_j}=\frac{1}{\sqrt{(d-1)!}}\sum_{\substack{{\pi}\in S_d\\{\pi}(1)=j}}\operatorname{sgn}({\pi})\ket{{\pi}(2)}\otimes\dots\otimes\ket{{\pi}(d)}\text{.}
\end{equation*}
\noindent $E:\mathcal{H}\to \mathcal{H}^{\otimes d-1}$ is an isometry, because the vectors $\{\ket{v_j}\}_{j\in[d]}$ in $\mathcal{H}^{\otimes d-1}$ are orthonormal. We switch back to the basis-labeling that starts with zero. Since $\{\ket{v_j}\}_{j\in\{0, 1,\dots d-1\}}$ are orthonormal, there exists a (non-unique) unitary $V_\text{conj}$  such that $V_\text{conj}\big(\ket{j}\otimes\ket{0}^{\otimes d-2}\big)=\ket{v_j}=E\ket{j}$, which implies that
\begin{equation}\label{eq_conj_alg}
    \left( I\otimes\bra{0}^{\otimes d-2}\right)V_\text{conj}^\dagger U^{\otimes d-1} V_\text{conj}\left(I\otimes\ket{0}^{\otimes d-2}\right)= C(U)
\end{equation}
\noindent is a unitary. Therefore, the dashed box $A_\text{conj}(U)=V_\text{conj}^\dagger U^{\otimes d-1} V_\text{conj}^\dagger$ in Fig. \ref{fig_Udagger} on input $\ket{\xi}\otimes\ket{0}^{\otimes d-2}$ outputs the all-zero state on the ancilla with probability
\begin{eqnarray*}
p&=&\left|\left|\left( I\otimes\bra{0}^{\otimes d-2}\right)A_\text{conj}(U)\left(\ket{\xi}\otimes\ket{0}^{\otimes d-2}\right)\right|\right|^2\\
&=&||C(U)\ket{\xi}||^2=1\text{.}
\end{eqnarray*}
\noindent In other words, projecting the ancilla output onto the all-zero state has no effect, i.e. $A_\text{conj}(U)( I\otimes\ket{0}\!\!\bra{0}^{\otimes d-2})$ equals
\begin{equation*}
   \left( I\otimes\ket{0}\!\!\bra{0}^{\otimes d-2}\right)A_\text{conj}(U)\left( I\otimes\ket{0}\!\!\bra{0}^{\otimes d-2}\right)\text{,}
\end{equation*}  
which by eq. \eqref{eq_conj_alg} equals to $ C(U)\otimes\ket{0}\!\!\bra{0}^{\otimes d-2}$.
\end{proof}

\subsection{Review of the simplest transpose and inverse algorithms}\label{section_transp_inv}

We listed the transpose and inverse task in \Cref{table_tasks2}.
Quintino, Dong, Shimbo, Soeda and Murao \cite{quintino_inverse, quintino_transforming} presented a family of algorithms for each task. Here we review only the simplest algorithm of each family; the one with the fewest queries and the lowest probability of success.

\begin{myconstr}[due to  \cite{quintino_inverse, quintino_transforming}]\label{thm_transpose}
Figure \ref{fig_Udagger} with the dashed box replaced by $U$ is a single-query algorithm implementing $U^T$ with the probability of success $1/d^2$. 
\end{myconstr}

\noindent Prepending to \Cref{thm_transpose} a unitary that prepares the input $\ket{\psi^+}=\tfrac{1}{\sqrt{d}}\sum_{i=0}^{d-1}\ket{ii}$ (the generalized Bell state) gives a postselection oracle algorithm exactly achieving the task $(U\mapsto U^T,\{id\})$. 

\begin{proof}
Note that if the dashed box is replaced by the identity operator $I$, Fig. \ref{fig_Udagger} corresponds to the generalized teleportation, which succeeds with probability $\frac{1}{d^2}$. In other words $\left(\bra{\psi^+}\otimes I\right)\left( I\otimes\ket{\psi^+}\right)=\frac{1}{d} I$. Next, for any complex $d\times d$ matrix $M$,
$\left(M\otimes I\right)\ket{\psi^+}= I\otimes M^{\,T} \ket{\psi^+}$,
which can be verified by writing $M$ in terms of its matrix elements. Thus, the algorithm of Fig. \ref{fig_Udagger} with the dashed box replaced by $U$ implements
\begin{eqnarray*}
    &\left(\bra{\psi^+}\otimes I\right)\left( I\otimes U\otimes I\right)\left( I\otimes\ket{\psi^+}\right)&\\
    &=U^T\left(\bra{\psi^+}\otimes I\right)\left( I\otimes\ket{\psi^+}\right)=\frac{1}{d}U^T&\text{.}\qedhere
\end{eqnarray*}
\end{proof}

Composing the transpose algorithm with the unitary algorithm for the complex conjugate \cite{compl_conj} gives a postselection oracle algorithm exactly achieving the task $(U\mapsto U^\dagger,\{id\})$.

\begin{myconstr}[due to  \cite{quintino_inverse, quintino_transforming}]
The algorithm of Fig. \ref{fig_Udagger} makes $d-1$ queries to $U$ and implements $U^\dagger$ with the probability of success $1/d^2$. 
\end{myconstr}

Both algorithms above have a useful feature: their effect on the ancillae can easily be uncomputed. Let $V_\text{prep}$ be a unitary that prepares $\ket{\psi^+}$ from the all-zero state. Substituting $\ket{\psi^+}=V_\text{prep}\ket{0}^{\otimes 2}$ into the algorithm of \Cref{fig_Udagger} for both occurrences of $\ket{\psi^+}$ gives algorithms with ancillae initialized to zeros and with $\Pi_\text{succ}=\ket{0}\!\!\bra{0}^{\otimes 2}$. Adding after the dashed box a gate that swaps the first and the last register results in transpose and inverse algorithms $A_T$ and $A_\dagger$ such that
\begin{eqnarray*}
    A_T(U)( I\otimes \ket{0}\!\!\bra{0}^{\otimes 2})&=&U^T\otimes \frac{1}{d}\ket{0}\!\!\bra{0}^{\otimes 2}\\
    A_\dagger(U)( I\otimes \ket{0}\!\!\bra{0}^{\otimes d})&=&U^\dagger\otimes \frac{1}{d}\ket{0}\!\!\bra{0}^{\otimes d}\text{.}
\end{eqnarray*}
We call such algorithms \emph{clean} (\Cref{def_clean}). Note the complex-conjugation algorithm, \cref{eq_complconj}, is also clean. \emph{Cleanness} can be achieved whenever the ancilla output is independent of the oracle $U$, by adding the appropriate fixed unitary gate as we did here.

\subsection{Query complexity of the clean transpose and inverse}\label{section_inv}

We reviewed the simplest transpose and inverse algorithms. References \cite{quintino_inverse, quintino_transforming} also presented other algorithms, with higher probabilities of success, that require more queries. Are they all clean? We answer this negatively in a much more general way: In this section we show that \emph{clean} transpose and inversion algorithms with certain numbers of queries are impossible. First we define clean algorithms: Apart from achieving their task (inequality \eqref{eq_eps-achieve}), they satisfy an additional condition regarding their effect on their ancillae -- if the ancilla input is the all-zero state, so should be the ancilla output. In the definition we merge the task and the ancilla conditions into one inequality. 

\begin{mydef}\label{def_clean}[Clean postselection oracle algorithm]
A postselection oracle algorithm $A:U(d)\to L(\mathcal{H}\otimes\mathcal{K})$ that $\epsilon$-approximately achieves the task $(t,\Sigma)$ is \emph{clean} if for all $\mathcal{H'}$ and  $\rho\in\mathcal{D}(\mathcal{H'}\otimes \mathcal{H})$, $\epsilon$ upper bounds also
\begin{equation*}
\left|\left|\frac{\widetilde{A(U)}\left(\rho\otimes\ket{\mathbf{0}}\!\!\bra{\mathbf{0}}\right)\widetilde{A(U)}{}^\dagger}{\operatorname{tr}\left[\widetilde{A(U)}\left(\rho\otimes\ket{\mathbf{0}}\!\!\bra{\mathbf{0}}\right)\widetilde{A(U)}{}^\dagger\right]}-\widetilde{t(U)}\rho \,\widetilde{t(U)}{}^\dagger\otimes\ket{\mathbf{0}}\!\!\bra{\mathbf{0}} \right|\right|_{\operatorname{tr}}\text{,}
\end{equation*}
where the wide tilde denotes the extension $\!\widetilde{\,X\;\,}\!\!:=I_{\mathcal{H}'}\otimes X$.
\end{mydef}

\noindent In other words, the ancilla output is independent of the unknown oracle. \emph{Clean} algorithms enable ancillae recycling and preserve coherences if used inside an interferometer (Fig. \ref{fig_mach-zehnder}). The following impossibility of certain clean algorithms is a corollary of \Cref{thm_neutr}.

\begin{mycor} Let $\epsilon\in [0, 1)$. If an $N$-query clean postselection oracle algorithm $\epsilon$-approximately achieves the task\label{thm_clean_inv_transp}
\setlength\topsep{0pt}
\begin{enumerate}[topsep=0pt,align=left, leftmargin=5pt] 
\setlength\parskip{0pt}
\setlength\itemsep{1pt}
    \item $(inv, \{id\})$, then $N \equiv -1 \pmod{d}$\label{item_inv}.
    \item $(U\mapsto U^T, \{id\})$, then $N \equiv 1 \pmod{d}$.\label{item_transp}
\end{enumerate}
\end{mycor}

\begin{proof}[Proof of \ref{item_inv}] 
Such an algorithm $(A, \{id\})$ satisfies
\begin{equation}
\left|\left|\frac{A(U)\ket{\mathbf{0}}\!\!\bra{\mathbf{0}}A(U)^\dagger}{\operatorname{tr}\left[A(U)\ket{\mathbf{0}}\!\!\bra{\mathbf{0}}A(U)^\dagger\right]}-(U^\dagger\otimes I)\ket{\mathbf{0}}\!\!\bra{\mathbf{0}}(U\otimes I)\right|\right|_{\operatorname{tr}}\leq\epsilon\label{eq_invclean}
\end{equation}
\noindent by \Cref{def_clean} with $\mathcal{H'}$ trivial and $\rho=\ket{0}\!\!\bra{0}$. If it makes $N$ queries, the new algorithm $({U\mapsto (U\otimes I)A(U)}, \{id\})$ makes $N+1$ queries. Represent by $||X||_{\operatorname{tr}}$ the left hand side of \eqref{eq_invclean}. Since $||(U\otimes I)X(U^\dagger\otimes I)||_{\operatorname{tr}}= ||X||_{\operatorname{tr}}$, the new algorithm $\epsilon$-approximately neutralizes the query sequence $\mathbf{s}=(id)^{N+1}$. \Cref{thm_neutr} then guarantees that $d$ divides $N+1$.\phantom{\qedhere}
\end{proof}

\begin{proof}[Proof of \ref{item_transp}]
Suppose some clean $N$-query algorithm $\epsilon$-approximately achieves the transpose task. Replace each query to $U$ by $A_\text{conj}(U)$ of \cref{eq_complconj}, i.e. by the clean, unitary, $(d-1)$-query algorithm of \cite{compl_conj} that exactly implements $\overline{U}$. The new algorithm makes $N(d-1)$ queries to $U$, $\epsilon$-approximately implements $U^\dagger$, and is clean. Then by the preceding proof $d$ divides $N(d-1)+1$, so $N \equiv 1 \pmod{d}$.
\end{proof}

References  \cite{quintino_inverse, quintino_transforming} found for every $k\in \mathbb{N}$ a $k$-query transpose algorithm and a $k(d-1)$-query inverse algorithm. The success probability of these grows with $k$. We have shown that their algorithms can be clean only if $k \equiv 1 \pmod{d}$.

\section{Discussion and Open Problems}\label{section_discussion}
In this work we introduce a unified framework for algorithms with unitary oracles $U\in U(d)$, and develop a topological method to prove their limitations. The method yields new limitations for algorithms implementing the if clause $c_\phi(U)$, neutralization, $1/d$th power $U^\frac{1}{d}$, transpose $U^T$ and inverse $U^\dagger$. Our method adds to the very few previously known proof methods for limiting query complexity.

\begin{figure*}
     \subfloat[]{\begin{minipage}[b]{0.4\textwidth}
         \includegraphics[valign=c]{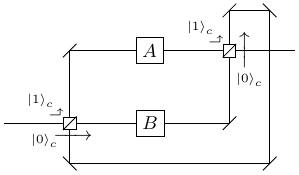}
         \label{fig_LO_QS}
        \end{minipage}
     }
     \subfloat[]{
     \begin{minipage}[b]{0.47\textwidth}
     \includegraphics[valign=c]{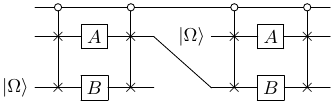}
     \vphantom{\includegraphics[valign=c]{fig8a_lin-opt_qswitch-cropped.pdf}}
         \label{fig_circuit_QS}
         \end{minipage}
     }
        \caption{(a) In this quantum switch implementation \protect{\cite{araujo2014computational}}, the photon passes through the interferometer twice. In the first pass, $A$ acts on one photon if the polarization is $\left|1\right>_c$ and on vacuum otherwise. In the second pass, $A$ acts on one photon if the polarization is $\left|0\right>_c$ and on vacuum otherwise. $A$ can act nontrivially on vacuum. (b) Analogously to \protect{\Cref{fig:two graphs}}, the equivalent circuit captures this possibility.
}
\end{figure*}

Similarly to the preceding methods \cite{lbounds_polynomials,bennett1997strengths,ambainis2000quantum,hoyer2007negative}, our method is worst-case: It expects algorithms to work on \emph{every} input. The distinction is important. While a quantum circuit implementing $U^\frac{1}{d}$ for most $U\in U(d)$ exists \cite{sheridan_maslov_mosca}, one working for all $U\in U(d)$ is impossible. Similarly, the worst-case if clause is impossible (\Cref{thm_main}), and the (exponential) average-case if clause is implemented by Fig. \ref{fig_cU_approx} after substituting the $U^\frac{1}{d}$ heuristic scheme \cite{sheridan_maslov_mosca} for $A(U)$. Whether a more efficient average-case algorithm exists remains open. Between the finite average-case and the infinite worst-case complexity, other complexity notions fit -- and could be explored. One notion is that of Levin's average-case complexity \cite{levin1986average}, equivalent to the notion of \emph{errorless} heuristic scheme \cite{bogdanov2006average} -- the scheme gives a warning (e.g. the $\perp$ symbol) if its output breaks the requested specifications (in our case, if the implemented operator exceeds the requested diamond distance from the task operator). Errorless heuristic schemes remain a possibility for the if clause and the fractional power. Another notion is that of smoothed complexity \cite{spielman2009smoothed}, which considers the worst input, but after perturbing it slightly. Though previously defined for errorless heuristic schemes \cite{spielman2009smoothed,blaser2015smoothed}, smoothed complexity can also be defined for heuristic schemes with errors (as ref. \cite{blum2002smoothed} did specifically for the perceptron algorithm). Such smoothed analysis could resolve the transition between the two problems' average-case and worst-case complexities. Focusing here on worst-case query complexity, we leave all the above questions open.

The query complexity bounds proven here reveal surprising differences between oracle computational models. First, we show how quantum circuits differ from linear optics. By \Cref{thm_main}, the if clause is an oracle problem on which quantum circuits perform much worse than linear optics. The query complexities of the if clause in these models are not polynomially, and not even exponentially equivalent. While the linear-optics complexity of the if clause was known to equal one, its quantum-circuit and process-matrix complexities are worked out here to be infinite! Linear optics avoid our impossibility because the direct sum $U\oplus I$ lacks homogeneity, a property required by our method. According to ref. \cite{araujo_cU}, the direct sum appears due to the gate's localization to the upper path-modes in Fig. \ref{fig_mach-zehnder}. However, instead of a direct sum, a second-quantization view yields a \emph{tensor product} of modes. We suggest a conceptually different reason for the direct sum: the \emph{linearity} of linear optics. Figure \ref{fig_mach-zehnder} implicitly assumes that $U$ is a linear (passive) gate: the vacuum $\ket{\Omega}$ is its $+1$ eigenstate. If this hidden assumption is made explicit, the oracle is not completely unknown and should be represented by $U'=1\oplus U$ as in Fig. \ref{fig_kitaev_vacuum}. 
Thus, while restricting to linear optics yields a weaker computational model, in the presence of oracles (\Cref{section_linear_optics}), it also imposes additional knowledge about the oracles. Surprisingly, overall this gives a remarkable advantage on the if clause.

The demonstrated advantage of linear optics motivates studying restricted computational models -- especially if the restriction makes the model more implementable in the near term. For example, besides linear optics, we could restrict optics to Gaussian optics or to various levels of non-Gaussianity \cite{chabaud2022holomorphic} -- applying the same restriction to the oracle. This could capture varying amounts of device-independence in experiments, adversarial power in cryptography, or distrust in communication protocols. Restricting the computational model -- and with it the oracle -- may change the query complexity of problems other than the if clause. Studying these changes under restrictions that are relevant in near-term computation is an interesting future direction. 

Process matrices compare equally unfavorably to linear optics on the if clause. While the linear restriction improves efficiency, relaxed causality does not.
This affects the debate of whether the linear optics experiments \cite{procopio2015experimental,goswami2018indefinite,rubino2017experimental} implement the process-matrix solution to the quantum switch: Any claimed linear-optics efficiency could stem from the implicit restriction on the oracles, and not from relaxed causality. To discount the restriction effect, we suggest to check the query complexity in the experiments by imagining that the oracles are fully general and act nontrivially on vacuum. For example, for the specific linear-optics quantum switch in \Cref{fig_LO_QS}, this leads to an equivalent circuit in \Cref{fig_circuit_QS} whose query complexity does not match the process-matrix quantum switch. This approach to query complexity has a specific, causality-related purpose. It might not work for other purposes. Indeed, in computational models that allow changing or superposing particle numbers, calculating complexity is an interesting and difficult open problem.

The last query complexity gap is between measurements. Understanding why process tomography fails at the if clause (it fails at the $(1/d)$th power for a similar reason) lead us to define modified tasks which are achievable: The \emph{entangled} and \emph{random if clause}.  If simpler implementations exist, the entangled if clause could serve as a subroutine in larger algorithms. The random if clause reveals the complexity gap between measurements: its query complexity is infinite in the one-success-outcome model, but finite (at most exponential) in a model with $\exp(n)$ success outcomes (\Cref{table_det-random}). Are there algorithms requiring fewer queries and fewer outcomes? An optimal algorithm will exactly quantify the advantage of many outcomes -- finding one remains an open question.

\begin{acknowledgments}
Z.G. and Y.T. thank Dorit Aharonov for the supervision and support. All authors thank Amitay Kamber for sketching the more elementary proof of \Cref{thm_top}, Mateus Ara\'ujo for enlightening discussions of the tomography issue, Alon Dotan, Itai Leigh, Giulio Gasbarri and Michalis Skotiniotis for useful feedback, and an anonymous referee from the QIP conference for suggesting an alternative approach via $SU(d)$ covering spaces (\Cref{section_cover_spaces} is our take on this approach). This work was supported by the Simons Foundation (Grant No. 385590), by the Israel Science Foundation (Grants No. 2137/19 and No. 1721/17), by the European Commission QuantERA grant 
  ExTRaQT (Spanish MICINN project No. PCI2022-132965), by the Ministry for Digital Transformation and of Civil Service of the Spanish Government through the QUANTUM ENIA project call, Quantum Spain project, by the European Union through the Recovery, Transformation and Resilience Plan - NextGenerationEU within the framework of the Digital Spain 2026 Agenda, by Spanish Agencia Estatal de Investigación (Project No. PID2022-141283NB-I00).
\end{acknowledgments}

\onecolumngrid
\appendix

\section{If-clause impossibility (special case) via the Borsuk-Ulam theorem}\label{section_bu}

Here we prove a special case of Theorem \ref{thm_main} from the famous result of algebraic topology: the Borsuk-Ulam theorem \cite{borsuk1933drei} (see Fig. \ref{fig_tides}).

\begin{butheorem}\label{thm_bu}
If $f:S^n\to \mathbb{R}^n$ is a continuous function, then there exists $x\in S^n$ such that $f(x)=f(-x)$. 
\end{butheorem}

\begin{proof}[Proof of Theorem \ref{thm_main} ($d$ even, $m$ odd)]
Let $d$ be an even and $m$ an odd integer. The first part is the same as in the general proof of \Cref{thm_main}. Assume towards contradiction that there exists a continuous $0$-homogeneous function $\mathcal{A}:U(d)\to \mathcal{S}_+(\mathbb{C}^{2d})$ such that 
\begin{equation}
    \left|\left|\frac{\mathcal{A}(U)(\rho)}{\operatorname{tr}\left[\mathcal{A}(U)(\rho)\right]}-\left(\left|0\right>\!\!\left<0\right|\otimes I + e^{i\phi(U)}\left|1\right>\!\!\left<1\right|\otimes U^m\right)\rho\, \left(\left|0\right>\!\!\left<0\right|\otimes I + e^{i\phi(U)}\left|1\right>\!\!\left<1\right|\otimes U^m\right)^\dagger\right|\right|<\frac{1}{2}\label{eq_bu_proof}
\end{equation}
holds for all $U\in U(d)$, all $\rho\in\mathcal{D}(\mathbb{C}^{2d})$ and for some real function $\phi$. Let $\rho_+=\ket{+}\!\!\bra{+}\otimes\left|0\right>\!\!\left<0\right|\in \mathbb{C}^2\otimes\mathbb{C}^d$, and define $f:U(d)\to \mathbb{C}$ by
\begin{equation}
    f(U):=\operatorname{tr}\left[\frac{\mathcal{A}(U)(\rho_+)}{\operatorname{tr}\left[\mathcal{A}(U)(\rho_+)\right]}\left(\left|1\right>\!\!\left<0\right|_c\otimes U^m\right)\right]\text{.}\label{def_f}
\end{equation}
This is well defined because $\mathcal{A}$ maps into $\mathcal{S}_+(\mathbb{C}^{2d})$ so the superoperator $\mathcal{A}(U)$ never vanishes the trace of its input. Observe that by its definition $f$ is a continuous odd function. Substituting  the inequality \eqref{eq_bu_proof} into \eqref{def_f}, we get that $|f(U)-\frac{1}{2}e^{-i\phi(U)}|<\frac{1}{2}$ so $f$ never maps to zero. We can define $\hat{f}(U)=f(U)/\sqrt{\overline{f(U)}f(U)}$, which is also continuous and odd, but maps into the circle $S^1$.

Now, since $d$ is even we can define the following continuous function $g:S^3\to U(d)$ such that given a vector $\mathbf{x}=(x_1, x_2, x_3, x_4)^T$ on the $3$-sphere $S^3\subset \mathbb{R}^4$
\begin{align*}
    g(\mathbf{x})&=\begin{pmatrix}
    x_1+ix_2&   & &  &   &-x_3+ix_4 \\
    & \ddots  &  &  &  \iddots & \\
    &   &x_1+ix_2 & -x_3+ix_4  &   & \\
    &   & x_3+ix_4  &x_1-ix_2& & \\
    & \iddots &   &   &\ddots & \\
    x_3+ix_4&   &   &   &   &x_1-ix_2
    \end{pmatrix}\text{,}
\end{align*}
where the dots denote that the elements are repeated across the diagonal and the antidiagonal, and there are zeroes everywhere else. Note that $g$ is odd, $g(-\mathbf{x})=-g(\mathbf{x})$. By the Borsuk-Ulam theorem, since $\hat{f}\circ g: S^3\to \mathbb{C}\subset  \mathbb{R}^3$ is continuous, there exists $\mathbf{x}\in S^3$ such that $\hat{f}\circ g(\mathbf{x})=\hat{f}\circ g(-\mathbf{x})=-\hat{f}\circ g(\mathbf{x})$, where the last equality holds because $\hat{f}\circ g$ is odd. So for this $\mathbf{x}$ we have $\hat{f}\circ g(\mathbf{x})=0$. But this contradicts the fact that $\hat{f}$ maps into $S^1$.
\end{proof}

\section{Exact if-clause impossibility: an operational proof}\label{section_operational}

\begin{proof}[Proof of \Cref{thm_main} ($\epsilon=0$)] Assume a postselection oracle algorithm $(A, \{id, inv\})$ exactly achieves the task $(c^m_\phi,  \{id, inv\})$. Since the query functions are continuous and $\pm 1$-homogeneous, by \cref{eq_general_alg} $A(U)$ is continuous and $w$-homogeneous for some $w\in\mathbb{Z}$. \Cref{eq_exactly-achieves} implies
\begin{equation}
A(U)\left(I\otimes\left|\mathbf{0}\right>\!\!\left<\boldsymbol{0}\right|\right)=\left(\left|0\right>\!\!\left<0\right|\otimes I+e^{i\phi(U)}\left|1\right>\!\!\left<1\right|\otimes U^m\right)\otimes\left|g(U)\right>\!\!\left<\boldsymbol{0}\right|\text{.}\label{eq_cU_achieved_exactly}
\end{equation}
\begin{figure}
         \includegraphics[]{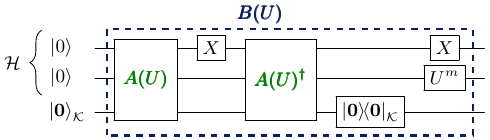}
         \caption{Circuit $B:U(d)\to L(\mathcal{H}\otimes\mathcal{K})$ running $A$ as a subroutine. The control qubit and the target qudit of $\mathcal{H}$ are shown explicitly. $X$ is a Pauli-$X$ gate.}
         \label{fig_proof}
\end{figure}
Consider the circuit of Fig. \ref{fig_proof}. It outputs $B(U)\left|\mathbf{0}\right>=e^{-i\phi(U)}\left|\left|\,\left|g(U)\right>\right|\right|^2\left|\mathbf{0}\right>$. Moreover, $B(U)$ is continuous and $m$-homogeneous. Then
\begin{equation*}
    f(U):=\frac{\left<\boldsymbol{0}\right|B(U)\left|\mathbf{0}\right>}{\left<\boldsymbol{0}\right|A(U)^\dagger A(U)\left|\mathbf{0}\right>}
\end{equation*}
is continuous, $m$-homogeneous and by \cref{condition_post} well-defined on all $U\in U(d)$. Observe that $f(U)=e^{-i\phi(U)}$. \Cref{thm_top} applies to $f$, so $d$ divides $m$. 
\end{proof}

The above can be extended to approximations, as long as the two sides of \eqref{eq_cU_achieved_exactly} are close enough to guarantee $f(U)\neq 0$.

\section{Exact if-clause impossibility via covering spaces}\label{section_cover_spaces}

This proof uses the fact that $PU(d)$ is a $d$th cover of $SU(d)$ \footnote{Suggested by an anonymous referee of the QIP conference, whom we would like to thank in this way.}. The covering space approach gives useful intuition, but, we did not find a way to extend it to the approximate case ($\epsilon>0$) -- we discuss the difficulty at the end of this section.

\begin{proof}[Proof of \Cref{thm_main} ($m=1$, $\epsilon=0$)]
Assume towards contradiction the existence of an algorithm $(\mathcal{A}, \Sigma_A)$ that $\epsilon$-approximately achieves $(c_\phi, \{id, inv\})$ for $\epsilon=0$. Observe from \cref{eq_postsel_superop,eq:process_matrix} that we can define a function on unitary \emph{superoperators} $\mathcal{A}_\mathcal{S}(\mathcal{U}):=\mathcal{A}(U)$, where the unitary superoperator $\mathcal{U}$ is defined as $\mathcal{U}(\rho):=U\rho U^\dagger$. Since the query alphabet $\Sigma_A\subseteq \{id, inv\}$ only contains continuous functions, $\mathcal{A}_\mathcal{S}(\mathcal{U})$ is continuous. Next define
\begin{equation}
    s(\mathcal{U}):=2d\frac{\left(\bra{1}\otimes I\right)\mathcal{A}_\mathcal{S}(\mathcal{U})(\ket{+}\!\!\bra{+}\otimes \tfrac{I}{d})\left(\ket{0}\otimes I\right)}{\operatorname{tr}\left[\mathcal{A}_\mathcal{S}(\mathcal{U})(\ket{+}\!\!\bra{+}\otimes \tfrac{I}{d})\right]}
\end{equation}
and observe that $s$ is continuous and that applying \cref{eq_0-achieve} we get $s(\mathcal{U})=e^{i\phi(U)}U$, which produces a contradiction with the lemma below.
\end{proof}

We will show that there is no continuous function $s$ (called \emph{section}) that maps a unitary superoperator to some matching unitary.

\begin{mylemma}\label{lem:cover1}
Let $\pi: U(d)\to \mathcal{S}_+(\mathbb{C}^d)$ map unitaries to the corresponding superoperators, i.e. $\pi: U\mapsto \mathcal{U}$, where $\mathcal{U}(\rho)=U\rho\, U^\dagger$.
There is no continuous map $s: Im(\pi)\to U(d)$ such that $\pi\circ s = id_{Im(\pi)}$. 
\end{mylemma}

First we will characterize the space $Im(\pi)$. Defining the equivalence relation $U\sim e^{i\alpha} U$ for all $\alpha\in[0, 2\pi)$, we will show that $Im(\pi)$ is homeomorphic to the quotient space $U(d)/\sim\,=U(d)/Z(U(d))=:PU(d)$ called the projective unitary space. In particular, $\pi$ maps into the same element precisely those unitaries that differ by the global phase. To see this, consider the continuous map $\pi$ and the quotient map $p:U(d)\to PU(d)$
\begin{center}
\begin{tikzpicture}[node distance=2cm, auto]
  \node (U) {$U(d)$};
  \node [right of=U] (im) {$Im\left(\pi\right)$};
  \node [node distance=1.2cm, below of=U] (PU) {$PU(d)$};
  \draw[->] (U) to node[above] {$\pi$} (im);
  \draw[->] (U) to node[left] {$p$} (PU);
  \draw[red, ->, text=red] (PU) to node {$h$} (im);
\end{tikzpicture}

\end{center}
and observe that whenever $p(U)=p(U')$ we have that $\pi(U)=\pi(U')$. By the universal property of quotient spaces in topology there exists a continuous map $h: PU(d)\to Im(\pi)$, $h([U])=\mathcal{U}$ such that $h\circ p=\pi$. Therefore $h$ is surjective (because $\pi$ is). $h$ is also injective: assume $h([U])=\mathcal{U}=\mathcal{U}'=h([U'])$ and let $U\in [U]$ be some unitary from the equivalence class. Choose coordinates $k, l\in [d]$ such that the matrix element $u_{lk}$ of $U$ is nonzero. Observing that $\sum_j\mathcal{U}\left(\ket{j}\!\!\bra{k}\right)\ket{l}\!\!\bra{j}=\overline{u}_{lk} U$ and similarly for some $U'\in [U']$, we get $\overline{u}_{lk} U=\overline{u'}_{lk} U'\neq 0$ and since both $U$ and $U'$ are unitary we have $U=e^{i\alpha}U'$ for some $\alpha\in[0, 2\pi)$. Therefore $[U]=[U']$ and $h$ is also injective. A bijective continuous function from a compact space to a Hausdorff space is a homeomorphism, so $Im(\pi)$ and $PU(d)$ are homeomorphic. We can restate \Cref{lem:cover1}:

\begin{mylemma}\label{thm_no-section}
 Let $p:U(d)\to PU(d)$ be the quotient map. There is no continuous map $s: PU(d)\to U(d)$ such that $p \circ s = id$. 
\end{mylemma}

We can prove this via \Cref{thm_top} (see \footnote{\emph{Proof.} $p$ is continuous and $0$-homogeneous, i.e. $p(e^{i\alpha}U)=p(U)$ for all $\alpha$ real. Suppose that $s$ as in \Cref{thm_no-section} exists, then the opposite composition $s\circ p$ must be continuous and $0$-homogeneous (as $p$ loses the information about $\alpha$). Also $(s\circ p)(U)$ differs from $U$ only by the global phase (because $p\circ s=id$). Define $f: U(d)\to S^1$, $f(U)=\bra{\mathbf{0}}U^{-1}(s\circ p)(U)\ket{\mathbf{0}}$. The function $f$ is continuous and $(-1)$-homogeneous so its existence contradicts \Cref{thm_top}. \hfill$\square$}). But our aim here is to prove \Cref{thm_no-section} directly, using the fact that $SU(d)$ is a finite cover of $PSU(d):=SU(d)/Z(SU(d))=SU(d)/\sim$.

\begin{proof}
First we show that $PU(d)$ and $PSU(d)$ are homeomorphic. Let $\imath:SU(d)\to U(d)$ be the embedding of the special unitaries into the unitaries and $sp$ and $p$ the quotient maps as depicted
\begin{center}
\begin{tikzpicture}[node distance=2cm, auto]
  \node (SU) {$SU(d)$};
  \node [right of=SU] (U) {$U(d)$};
  \node [node distance=1.2cm, below of=SU] (PSU) {$PSU(d)$};
  \node [node distance=1.2cm, below of=U] (PU) {$PU(d)$};
  \draw[right hook->] (SU) to node[above] {$\imath$} (U);
  \draw[->] (SU) to node[left] {$sp$} (PSU);
  \draw[->] (U) to node[left] {$p$} (PU);
  \draw[red, ->, text=red] (PSU) to node[above] {$g$} (PU);
\end{tikzpicture}
\end{center}
Observing that whenever $sp(U)=sp(U')$ we must have that $U=\lambda U'$ (with $\lambda^d=1$) and therefore $(p\circ \imath)(U)=(p\circ \imath)(U')$, we get by the universal property of quotient spaces in topology that there exists a continuous $g:PSU(d)\to PU(d)$ that makes the diagram commute. Next define the function $r: U(d)\to SU(d)$, $r(U)=U/\operatorname{det}(U)$. This is a continuous function such that $r\circ \imath=id_{SU(d)}$ (a retraction).
\begin{center}
\begin{tikzpicture}[node distance=2cm, auto]
  \node (SU) {$SU(d)$};
  \node [right of=SU] (U) {$U(d)$};
  \node [node distance=1.2cm, below of=SU] (PSU) {$PSU(d)$};
  \node [node distance=1.2cm, below of=U] (PU) {$PU(d)$};
  \draw[<-] (SU) to node[above] {$r$} (U);
  \draw[->] (SU) to node[left] {$sp$} (PSU);
  \draw[->] (U) to node[left] {$p$} (PU);
  \draw[red, <-, text=red] (PSU) to node[above] {$g'$} (PU);
\end{tikzpicture}
\end{center}
Whenever $p(U)=p(U')$ we have that $U=e^{i\alpha}U'$ for some $\alpha\in[0,2\pi)$ and therefore $(sp\circ r)(U)=(sp\circ r)(U')$. By the universal property of quotient spaces in topology there exists a continuous $g':PU(d)\to PSU(d)$ that makes the diagram commute. The following argument shows that $g$ and $g'$ are the inverse of each other: Take any $\bar{x}\in PSU(d)$ and some preimage of it $x\in SU(d)$, $sp(x)=\bar{x}$. Observe that
\begin{equation*}
    g'\circ g\,(\bar{x})=g'\circ g\circ sp\,(x)=g'\circ p\circ \imath\,(x)=sp\circ r \circ \imath\,(x)=sp(x)=\bar{x}\text{,}
\end{equation*}
where the second equality follows from the first commuting diagram, the third equality from the second commuting diagram and the fourth equality from the fact that $r$ is a retraction. Similarly take any $[U]\in PU(d)$. For any preimage $U\in U(d)$ of $[U]$, $p(U)=[U]$ we have
\begin{equation*}
    g\circ g'\,([U])=g\circ g'\circ p\,(U)=g\circ sp\circ r\,(U)=p\circ \imath \circ r\,(U) = [U]\text{,}
\end{equation*}
where the second equality is from the second commuting diagram, the third equality is from the first commuting diagram, and the last equality holds because $p$ maps $U/\det(U)$ and $U$ to the same equivalence class. We got that the continuous maps $g$ and $g'$ are the inverse of each other, so $PSU(d)$ and $PU(d)$ are homeomorphic. In particular we can find the fundamental group of $PU(d)$ by calculating it for $PSU(d)$.

Next we use the covering space argument to calculate $\pi_1(PSU(d))$, the fundamental group of $PSU(d)$. The group of $d$th roots of unity (same group as $\mathbb{Z}/d\mathbb{Z}$) acts on $SU(d)$ by scalar multiplication: $\lambda.U=\lambda U$. Since the group of $d$th roots of unity is finite, its action is a \emph{covering space action}: for every $A\in SU(d)$ there is a neighborhood $W$ of $A$ such that if for some $d$th root of unity $\lambda$ the neighborhood $\lambda.W$ intersects $W$ then $\lambda=1$. Then the covering space theorem (Proposition 1.40 in  \cite{hatcher2000algebraic}) states that $sp:SU(d)\to PSU(d)$ is a covering map and that the group of $d$th roots of unity $\mathbb{Z}/d\mathbb{Z}$ is isomorphic to $\pi_1(PSU(d))/sp_{*}(\pi_1(SU(d)))$. We know that $sp_{*}(\pi_1(SU(d)))=\{1\}$ because the induced homomorphism $sp_{*}$ must map the trivial group to the trivial group and $\pi_1(SU(d))$ is indeed trivial because $SU(d)$ is simply connected. We get that $\pi_1(PSU(d))$ is isomorphic to $\mathbb{Z}/d\mathbb{Z}$ and so is $\pi_1(PU(d))$.

Now assume towards contradiction that $s:PU(d)\to U(d)$ as in \Cref{thm_no-section} exists. In other words we have the following commuting diagram
\begin{center}
\begin{tikzpicture}[node distance=2cm, auto]
  \node (PU1) {$PU(d)$};
  \node [right of=PU1] (U) {$U(d)$};
  \node [right of=U] (PU2) {$PU(d)$};
  \draw[->] (PU1) to node[above] {$s$} (U);
  \draw[->] (U) to node[above] {$p$} (PU2);
  \draw[->, bend right] (PU1) to node[above] {$id$} (PU2);
\end{tikzpicture}
\end{center}
and knowing that $\pi_1(PU(d))=\mathbb{Z}/d\mathbb{Z}$ and that $\pi_1(U(d))=\mathbb{Z}$, we have the corresponding diagram for the induced homomorphisms between the fundamental groups:
\begin{center}
\begin{tikzpicture}[node distance=2cm, auto]
  \node (PU1) {$\mathbb{Z}/d\mathbb{Z}$};
  \node [right of=PU1] (U) {$\mathbb{Z}$};
  \node [right of=U] (PU2) {$\mathbb{Z}/d\mathbb{Z}$};
  \draw[->] (PU1) to node[above] {$s_*$} (U);
  \draw[->] (U) to node[above] {$p_*$} (PU2);
  \draw[->, bend right] (PU1) to node[above] {$id$} (PU2);
\end{tikzpicture}
\end{center}
\vspace{-12pt}
Since $s_*$ is a homomophism, it must map the generator of $\mathbb{Z}/d\mathbb{Z}$ to some $x\in \mathbb{Z}$ such that $\overbrace{x+x+\dots+x}^{d \text{ times}}=0$. Therefore $s_*$ must map the generator -- and consequently any element -- of $\mathbb{Z}/d\mathbb{Z}$ to zero. $s_*$ must be trivial and this contradicts $p_*\circ s_*=id$.
\end{proof}

From the above proofs, extensions to approximations seem nontrivial. We reduced any if-clause algorithm to a map $s$ of \Cref{thm_no-section} that maps an equivalence class $[U]$ to some representative from that class $U_\text{rep}\in [U]$. But if the algorithm is only approximate, then so is $s$ obtained from the reduction. Its output could be outside the input equivalence class, $s([U])\notin [U]$ or $p\circ s\neq id$. To reach contradiction we would need an extension of \Cref{thm_no-section} to approximations $p\circ s\approx id$.

\section{Alternative proof of \Cref{thm_top}}\label{section_top_short}
Central to this proof is the fact that the determinant
map $\det:U\left(d\right)\rightarrow S^{1}$ induces an isomorphism
on the corresponding fundamental groups (Proposition 2.2.6 of  \cite{symplectic}). The required algebraic-topology background, all of which
is covered in Chapter 9 of  \cite{munkres}, is just the basics of fundamental groups (paths,
homotopies, time reparametrization etc.) together with the description
of the fundamental group of the circle $S^{1}$ as infinite cyclic
and generated by the homotopy class of the loop $t\mapsto e^{2\pi it}.$
Throughout the proof, we will use the notation $\pi_{1}\left(X,p\right)$
for the fundamental group of the topological space $X$ based at $p$.
For a path $\gamma:\left[0,1\right]\rightarrow X$, we will denote
by $\left[\gamma\right]$ its homotopy class, and for another path
$\gamma'$ such that $\gamma\left(1\right)=\gamma'\left(0\right)$
we will denote by $\gamma*\gamma'$ their composition. 
\begin{mylemma}
\label{lem: loop is d-th power}Let $\gamma:\left[0,1\right]\rightarrow U\left(d\right)$
be the loop $\gamma\left(t\right)=e^{2\pi it} I$, where $ I\in U\left(d\right)$
is the identity matrix. Then its homotopy class $\left[\gamma\right]$
is a $d$th power in the fundamental group $\pi_{1}\left(U\left(d\right), I\right)$,
i.e. there exists some other loop $\gamma':\left[0,1\right]\rightarrow U\left(d\right)$
based at $ I$ such that $\left[\gamma'\right]^{d}=\left[\gamma\right]$.
\end{mylemma}

We proved \Cref{lem: loop is d-th power} (less formally) in the main text. Here we give another, shorter proof

\begin{proof}[Proof of \Cref{lem: loop is d-th power}] We use the fact (see Proposition 2.2.6 of  \cite{symplectic}) that the determinant
map $\det:U\left(d\right)\rightarrow S^{1}$ induces an isomorphism on the fundamental groups $\det_{*}:\pi_{1}\left(U\left(d\right), I\right)\rightarrow\pi_{1}\left(S^{1},1\right)$.
The fundamental group $\pi_{1}\left(S^{1},1\right)$ is isomorphic
to $\mathbb{Z}$ and generated by the homotopy class of the loop $\nu:\left[0,1\right]\rightarrow S^{1}$
defined by $\nu\left(t\right)=e^{2\pi it}.$ Since $\det\circ\gamma:\left[0,1\right]\rightarrow S^{1}$
is the loop $t\mapsto e^{2\pi idt}$, which winds $d$ times around
the circle, it is homotopic to the composition of $\nu$ with itself
$d$ times (as both loops are equal up to reparametrization of time).
Thus,
\begin{equation}
{\det}_*\left(\left[\gamma\right]\right)=\left[\det\circ\,\gamma\right]=\left[\underbrace{\nu*\nu*...*\nu}_{d\ \text{times}}\right]=\underbrace{\left[\nu\right]\left[\nu\right]\dots\left[\nu\right]}_{d\ \text{times}}=\left[\nu\right]^{d}.\label{eq_dth_power}
\end{equation}
Since $\det_{*}\left(\left[\gamma\right]\right)$ is a $d$th power
in $\pi_{1}\left(S^{1},1\right)$ and $\det_{*}$ is a group isomorphism,
we deduce that $\left[\gamma\right]$ is also a $d$th power (to
see this, apply the inverse of $\det_{*}$ to both sides of
\cref{eq_dth_power}).
\end{proof}

To complete the proof of \Cref{thm_top} we give here the formal version of the arguments in the main text.

\begin{proof}[Proof of \Cref{thm_top}]
Without the loss of generality we may assume that $f\left( I\right)=1$,
since otherwise we may replace $f$ with its rotation $U\mapsto f\left( I\right)^{-1}f\left(U\right)$.
Let $\gamma:\left[0,1\right]\rightarrow U\left(d\right)$ be the loop
$\gamma\left(t\right)=e^{2\pi it} I$. By \Cref{lem: loop is d-th power},
the homotopy class $\left[\gamma\right]$ is a $d$th power in $\pi_{1}\left(U\left(d\right), I\right)$,
and so must map to a $d$th power under any group homomorphism. In
particular, $h_{*}\left(\left[\gamma\right]\right)=\left[f\circ\gamma\right]$
is a $d$th power in $\pi_{1}\left(S^{1},1\right)$, and so we can
write $\left[f\circ\gamma\right]=\left[\gamma'\right]^{d}$ for some
loop $\gamma':\left[0,1\right]\rightarrow S^{1}$. Since the fundamental
group of the circle $\pi_{1}\left(S^{1},1\right)$ is (infinite) cyclic
and generated by the homotopy class of the loop $\nu\left(t\right)=e^{2\pi it}$,
we can write $\left[\gamma'\right]=\left[\nu\right]^{k}$ for some
$k\in\mathbb{Z}$, and so we have $\left[f\circ\gamma\right]=\left[\nu\right]^{kd}$.
On the other hand, the $m$-homogeneity of $f$ allows us to derive
an explicit formula for $f\circ\gamma$:
\[
f\circ\gamma\left(t\right)=f\left(e^{2\pi it} I\right)=\left(e^{2\pi it}\right)^{m} f\left( I\right)=e^{2\pi imt}.
\]
The loop $t\mapsto e^{2\pi imt}$, which winds $m$ times around the
circle, is homotopic to the composition of $\nu$ with itself $m$
times (as both loops are equal up to reparametrization of time) and
so $\left[f\circ\gamma\right]=\left[\nu\right]^{m}.$ Combining the
two expressions we got for $\left[f\circ\gamma\right]$ in terms of
$\left[\nu\right]$, we deduce that $\left[\nu\right]^{kd}=\left[\nu\right]^{m}$.
Since the order of $\left[\nu\right]$ in $\pi_{1}\left(S^{1},1\right)$
is infinite, the two powers must be equal and $m=kd$.
\end{proof}

\section{Process tomography of \Cref{def_tom2} exists}
\label{section_def_tom2_exist}

In this appendix we argue that algorithms of \Cref{def_tom2} exist. For any $U\in U(d)$, they output a classical description of a unitary operator close to $\kappa_r(U)$, where $r$ depends on a measurement outcome, so $r$ is random. In example \eqref{eq:kappa_cases_appendix} $r\in\{0,1,\dots d-1\}$. The classical description allows for building a new circuit that is close to implementing $\ket{0}\!\!\bra{0}\otimes I + \ket{1}\!\!\bra{1}\otimes \kappa_r(U)$. This does not approximately achieve the if clause, and the only reason is the randomness of $r$. Intuitively, it achieves a random relaxation of the if clause, discussed rigorously in \Cref{section_relaxed}.

Known process tomography algorithms achieve
\begin{equation}
    \operatorname{Pr}_{X\sim p_{N,U}(\cdot)}\left[\left|\left|X\cdot X^\dagger-U\cdot U^\dagger\right|\right|_\diamondsuit\leq \epsilon_{\diamondsuit N}\right]\geq 1-\delta_{\diamondsuit N}\text{,}\label{eq:tom_diam}
\end{equation}
where $\lim_{N\to\infty}\epsilon_{\diamondsuit N}=0$ and  $\lim_{N\to\infty}\delta_{\diamondsuit N}=0$ \cite{surawy2022projected}. 

\begin{mytheorem}
Any process tomography algorithm that satisfies \eqref{eq:tom_diam} also satisfies \Cref{def_tom2}.
\end{mytheorem}

\begin{proof}
Suppose that $\left|\left|X\cdot X-U\cdot U^\dagger\right|\right|_\diamondsuit\leq \epsilon_{\diamondsuit N}<1/d$. By Theorem 1.3 of \cite{technical} there exists $\alpha_{U,X}\in\mathbb{C}$,
\begin{equation}
    (1-\epsilon_{\diamondsuit N})^2\leq 1-{\epsilon_{\diamondsuit N}}\leq \left|\alpha_{U,X}\right|^2 \leq 1\text{,}\label{eq:alpha_phase}
\end{equation}
such that
\begin{equation}
   \left|\left|X- \alpha_{U,X}U\right|\right|_{\operatorname{op}}\leq 2\sqrt {\epsilon_{\diamondsuit N}}\text{.}\label{eq:op_from_tech}
\end{equation}

\begin{samepage}
\noindent This implies that the inequality
\begin{align}
    \left|\bra{r}X\ket{0}-\alpha_{U,X}\bra{r}U\ket{0}\right|&\leq 2\sqrt {\epsilon_{\diamondsuit N}}\label{eq:ineq}
\end{align}
holds for any $r\in\{0,1,\dots, d-1\}$.
\end{samepage}

We aim to satisfy \Cref{def_tom2} with the $\kappa_j$ functions and the $r=r(X)$ rule described in \cref{eq:kappa_ex} and below it:
\begin{align}
    \kappa_j(U)&=  
    \begin{cases}
    \frac{\overline{{\langle j |U|0\rangle}}}{|\langle j |U|0\rangle|}U = \frac{|\langle j |U|0\rangle|}{\langle j |U|0\rangle}U & \text{if } \langle j |U|0\rangle\neq 0 \\
    \kappa_{j+1 \operatorname{mod}d}(U) & \text{otherwise}
    \end{cases}\label{eq:kappa_cases_appendix}\\
    r&=r(X)=\min\{j: |\langle j |X|0\rangle|\geq 1/\sqrt{d}\}\text{,}
\end{align}

\noindent which ensures that $|\langle r |X|0\rangle|\geq1/\sqrt{d}$ and because of our supposition also that $\langle r |U|0\rangle\neq 0$, so we can safely divide by each. We will show that the left-hand side of \eqref{eq:op_from_tech} is close to $\left|\left|\kappa_r(X)-\kappa_r(U)\right|\right|_{\operatorname{op}}$, by using the following upper bound:
\begin{align}
    \left|\alpha_{U,X}- \tfrac{|\langle r |U|0\rangle|}{|\langle r |X|0\rangle|}\tfrac{\bra{r}X\ket{0}}{\bra{r}U\ket{0}}\right|
    &\leq 
    \left|\alpha_{U,X}-\tfrac{\alpha_{U,X}}{|\alpha_{U,X}|}\right|
    +
    \left|\tfrac{\alpha_{U,X}}{|\alpha_{U,X}|} - \alpha_{U,X}\tfrac{|\langle r |U|0\rangle|}{|\langle r |X|0\rangle|}\right|
    +
    \left|\alpha_{U,X}\tfrac{|\langle r |U|0\rangle|}{|\langle r |X|0\rangle|} - \tfrac{|\langle r |U|0\rangle|}{|\langle r |X|0\rangle|}\tfrac{\bra{r}X\ket{0}}{\bra{r}U\ket{0}}\right|\nonumber\\
    &= \left||\alpha_{U,X}|-\tfrac{|\alpha_{U,X}|}{|\alpha_{U,X}|}\right|
    +
    |\alpha_{U,X}| \left|\tfrac{1}{|\alpha_{U,X}|} - \tfrac{|\langle r |U|0\rangle|}{|\langle r |X|0\rangle|}\right|
    +
    \tfrac{|\langle r |U|0\rangle|}{|\langle r |X|0\rangle|}\left|\alpha_{U,X} - \tfrac{\bra{r}X\ket{0}}{\bra{r}U\ket{0}}\right|\nonumber\\
    &\leq \left|\alpha_{U,X}-1\right| + |\alpha_{U,X}|\frac{2\sqrt {\epsilon_{\diamondsuit N}}}{|\alpha_{U,X}|\,|\langle r |X|0\rangle|} + \frac{|\langle r |U|0\rangle|}{|\langle r |X|0\rangle|}\frac{2\sqrt {\epsilon_{\diamondsuit N}}}{|\langle r |U|0\rangle|}\label{eq:last_line}\text{,}
\end{align}
\noindent where the last inequality upper bounded: the first term using $|\,|a| - |b|\,|\leq |a-b|$ $a,b\in\mathbb{C}$, the second term using the same applied to \cref{eq:ineq} with $\alpha_{U,X}\langle r |X|0\rangle$ taken out of the norm, the third term using \cref{eq:ineq} with $\langle r |U|0\rangle$ taken out. Applying \eqref{eq:alpha_phase} to the first term of \eqref{eq:last_line} and simplifying the rest, we get
\begin{equation}
    \left|\alpha_{U,X}- \frac{|\langle r |U|0\rangle|}{|\langle r |X|0\rangle|}\frac{\bra{r}X\ket{0}}{\bra{r}U\ket{0}}\right|\leq \epsilon_{\diamondsuit N} + \frac{4\sqrt {\epsilon_{\diamondsuit N}}}{|\langle r |X|0\rangle|}\label{eq:scalars_close}\text{,}
\end{equation}

\noindent We obtain
\begin{align}
\left|\left|\kappa_r(X)-\kappa_r(U)\right|\right|_{\operatorname{op}} 
&=\left|\left|\frac{|\langle r |X|0\rangle|}{\bra{r}X\ket{0}}X-\frac{|\langle r |U|0\rangle|}{\bra{r}U\ket{0}}U\right|\right|_{\operatorname{op}} \nonumber\\
&=\left|\left|X-\frac{|\langle r |U|0\rangle|}{|\langle r |X|0\rangle|}\frac{\bra{r}X\ket{0}}{\bra{r}U\ket{0}}U\right|\right|_{\operatorname{op}}\nonumber\\
&\leq\left|\left|X-\alpha_{U,X}U\right|\right|_{\operatorname{op}} +\left|\left|\alpha_{U,X}U-\frac{|\langle r |U|0\rangle|}{|\langle r |X|0\rangle|}\frac{\bra{r}X\ket{0}}{\bra{r}U\ket{0}}U\right|\right|_{\operatorname{op}}\nonumber\\
&\leq 2\sqrt {\epsilon_{\diamondsuit N}} + \epsilon_{\diamondsuit N} + 4\frac{\sqrt {\epsilon_{\diamondsuit N}}}{|\langle r |X|0\rangle|}\label{eq:op_upper_bounded_r}
\end{align}

\noindent The last inequality upper bounds the first term by \eqref{eq:op_from_tech} and the second term by \eqref{eq:scalars_close}. 

Provided that condition \eqref{eq:tom_diam} holds, we set the parameters of \Cref{def_tom2} in the following way. Let $\delta_N=\delta_{\diamondsuit N}$. Let $N_0$ be some number such that $\epsilon_{\diamondsuit {N}}\leq 1/d$ for all $N\geq N_0$. We set

\begin{equation}
    \epsilon_{N}=
    \begin{cases}
    2 & \text{\quad if } N<N_0\\
    2\sqrt {\epsilon_{\diamondsuit N}} + \epsilon_{\diamondsuit N} + 4\sqrt{d}\sqrt {\epsilon_{\diamondsuit N}} & \text{\quad if } N\geq N_0\text{.}
    \end{cases}
\end{equation}
The $r=r(X)$ rule guarantees that for $N\geq N_0$ $\epsilon_N$ upper bounds \eqref{eq:op_upper_bounded_r}. At the same time $\lim_{N\to\infty}\epsilon_N=0$ and  $\lim_{N\to\infty}\delta_N=0$ hold, satisfying \Cref{def_tom2}.
\end{proof}

\section{Random if clause and entangled if clause}\label{section_relaxed}

Here we show how algorithms of \Cref{def_tom2} yield the random if clause and the entangled if clause. For a rigorous treatment we first need several definitions.

\begin{mydef}[$k$-task]\label{def:k-task}
A \emph{$k$-task} is a pair $(\boldsymbol{t}, \Sigma)$, where
\begin{enumerate}
    \item {$\boldsymbol{t}$ is a vector $(t_0, t_1,\dots t_{k-1})$ of task functions $t_i: U(d)\to L(\mathcal{H}_{\boldsymbol{t}})$}
    \item {The query alphabet $\Sigma$ is a set of query functions.}
\end{enumerate}
\end{mydef}

\noindent For example, the $k$-if clause is $((c_{\phi_0}, c_{\phi_1},\dots,c_{\phi_{k-1}}), \{id, inv\})$ for any $\{\phi_i\}_i$. Another example is the $k$-$q$th power for some fraction $q\in\mathbb{Q}$, where each $t_i$ is a specific $q$th power function on $U\in U(d)$ (recall that many such functions exist).

A random task or an entangled task is a $k$-task for any $k\geq 2$. The qualifiers “random” and “entangled” will invoke different notions of algorithms “achieving” the $k$-task. We defined postselection oracle algorithm (\Cref{def_postsel_alg}), which includes only a binary success/fail measurement. Here we need a generalization:

\begin{mydef}[$m$-Postselection oracle algorithm]\label{def_m-postsel_alg} An \emph{$m$-postselection oracle algorithm} $((A_0, A_1, \dots A_{m-1}), \Sigma_A)$ corresponds to an $m$-tuple of functions $A_r:U(d)\to L(\mathcal{H}\otimes\mathcal{K})$ of the form
\begin{equation}\label{eq_general_m-alg}
    A_r(U) = \Pi_r\, V_{N} (\sigma_N(U)\otimes I_{ \mathcal{K}_N}) \dots V_2 (\sigma_2(U)\otimes I_{ \mathcal{K}_2})V_1(\sigma_1(U)\otimes I_{ \mathcal{K}_1})V_0 
\end{equation}
that differ from each other only by the last multiplication by the projector $\Pi_r$. The operator $A_r(U)$ is implemented upon getting the success outcome $r\in\{0, 1,\dots m-1\}$ in the measurement  $\{\Pi_0, \Pi_1, \dots \Pi_{m-1}, \Pi_\text{fail}\}$, where $\Pi_\text{fail}=I-\sum_{r=0}^{m-1}\Pi_r$. As before, $\sigma_i\in\Sigma_A$. Whenever the ancilla Hilbert space $\mathcal{K}$ is initialized to the all-zero state, the total probability of success is nonzero, i.e. for all $U\in U(d)$ and for all $\ket{\xi}\in\mathcal{H}$
\begin{equation}
    \sum_{r=0}^{m-1}||A_r(U)\left(\ket{\xi}\otimes\ket{\mathbf{0}}_ \mathcal{K}\right)||^2> 0\text{.}
    \label{condition_m-post}
\end{equation}
\end{mydef}

\noindent Observe that for $m=1$ we are back to our original postselection oracle algorithm model. The superoperator implemented on the $r$th outcome is
\begin{equation}
    \mathcal{A}_{r\,U}(\rho)=\operatorname{tr}_{\mathcal{K}}\left[A_r(U)\left(\rho\otimes\ket{\mathbf{0}}\!\!\bra{\mathbf{0}}_\mathcal{K}\right)A_r(U)^\dagger\right]\label{eq_superops}
\end{equation}
where we wrote $U$ as a subscript to avoid too many parentheses. If $m\geq2$, the probability $\operatorname{tr}[\mathcal{A}_{r\,U}(\rho)]$ of an individual outcome $r\in\{0, 1,\dots m-1\}$ can hit zero.

Next, we generalize approximately achieving (\Cref{def_supermap_approx_achieving}) to random $k$-tasks. We would like to compare $\mathcal{A}_{r\,U}$ to $t_r(U)$ as in \Cref{def_supermap_approx_achieving}. The $r$th contribution to the error should be weighted by the conditional probability of the $r$th outcome, $p_{r\,U}(\rho)=\operatorname{tr}[\widetilde{\mathcal{A}_{r\,U}}(\rho)]/p_{\text{succ}\,U}(\rho)$, where $\rho\in\mathcal{D}(\mathcal{H}\otimes\mathcal{H}')$, $\widetilde{\mathcal{A}_{r\,U}}=\mathcal{A}_{r\,U}\otimes\mathcal{I}_{\mathcal{H}'}$, and $p_{\text{succ}\,U}(\rho)=\sum_{r=0}^{m-1}\operatorname{tr}[\widetilde{\mathcal{A}_{r\,U}}(\rho)]$ is the total probability of success. On inputs $U,\rho$ that yield a nonzero probability of outcome $r$, the contribution is
\begin{equation}
    p_{r\,U}(\rho)\left|\left|\frac{\widetilde{\mathcal{A}_{r\,U}}(\rho)}{\operatorname{tr}\left[\widetilde{\mathcal{A}_{r\,U}}(\rho)\right]}-t_r(U)\rho \,t_r(U)^\dagger\right|\right|_{\operatorname{tr}} = 
    \frac{1}{p_{\text{succ}\,U}(\rho)}\left|\left|\widetilde{\mathcal{A}_{r\,U}}(\rho)-\operatorname{tr}\left[\widetilde{\mathcal{A}_{r\,U}}(\rho)\right]t_r(U)\rho \,t_r(U)^\dagger\right|\right|_{\operatorname{tr}}\text{,}
\end{equation}
which is well-defined on all $\rho\in\mathcal{D}(\mathcal{H})$ because of condition \eqref{condition_m-post}. This justifies the following definition.

\begin{mydef}[$m$-postselection oracle algorithm $\epsilon$-approximately achieving a random $k$-task]\label{def_achieving_rand}
 \hfill\newline
An $m$-postselection oracle algorithm $((A_0, A_1, \dots A_{m-1}), \Sigma_A)$ {\it $\epsilon$-approxi\-mately achieves} a random $((t_0, t_1,\dots t_{k-1}), \Sigma)$ if $\Sigma_A\subseteq \Sigma$ and if for all $U\in U(d)$
\begin{equation}
    \sup_{\substack{\mathcal{H}'\\[0.3em]\rho\in\mathcal{D}(\mathcal{H}\otimes\mathcal{H}')\hspace{-1.7em}}}\frac{1}{p_{\text{succ}\,U}(\rho)}\sum_{r=0}^{m-1}\left|\left|\widetilde{\mathcal{A}_{r\,U}}(\rho)-\operatorname{tr}\left[\widetilde{\mathcal{A}_{r\,U}}(\rho)\right]\widetilde{t_r(U)}\rho \,\widetilde{t_r(U)}^\dagger\right|\right|_{\operatorname{tr}}\leq \epsilon\text{,} \label{eq_eps-achieve_many} 
\end{equation}
\noindent where the wide tilde denotes the $\mathcal{H}'$ extension by the identity superoperator or operator, and for $r\geq k$ we set $t_r=t_{k-1}$.
\end{mydef}

If either $k=1$ or $m=1$, then \eqref{eq_eps-achieve_many} simplifies to \eqref{eq_eps-achieve}. Consequently, the if clause ($k=1$) is impossible also in the $m>1$ model, and the random if clause is impossible in the original $m=1$ model. This justifies the remaining impossibilities in \Cref{table_det-random}.

Extending the algorithmic model is not necessary for the entangled if clause. Binary success/fail measurement suffices if the $r$ value can remain in an additional, entangled register.  

\begin{mydef}[Postselection oracle algorithm $\epsilon$-approximately achieving an entangled $k$-task]\label{def_achieving_ent}
A postselection oracle algorithm $(\mathcal{A}, \Sigma_A)$ {\it $\epsilon$-approxi\-mately achieves} an entangled $((t_0, t_1,\dots t_{k-1}), \Sigma)$ if $\Sigma_A\subseteq \Sigma$ and if $\mathcal{A}_U:=\mathcal{A}(U)$ maps $L(\mathcal{H})\to L(\mathcal{H}\otimes\mathbb{C}^{k}$) and
\begin{equation}
    \sup_{\substack{\mathcal{H}'\\[0.3em]\rho\in\mathcal{D}(\mathcal{H}\otimes\mathcal{H}')\hspace{-1em}}}\frac{1}{\operatorname{tr}\left[\widetilde{\mathcal{A}_U}(\rho)\right]}\sum_{r=0}^{k-1}\left|\left|\bra{r}\widetilde{\mathcal{A}_U}(\rho)\ket{r}-\operatorname{tr}\left[\bra{r}\widetilde{\mathcal{A}_U}(\rho)\ket{r}\right]\widetilde{t_r(U)}\rho \,\widetilde{t_r(U)}^\dagger\right|\right|_{\operatorname{tr}}\leq \epsilon\text{,}\label{eq:achieving_ent}
\end{equation}
where $\{\ket{r}\}_{r\in [k]}$ is the computational basis of $\mathbb{C}^k$, and the wide tilde denotes the $\mathcal{H}'$ extension by the identity superoperator or operator.
\end{mydef}

\noindent Writing $p_{r\, U}(\rho)=\operatorname{tr}\left[\bra{r}\widetilde{\mathcal{A}_U}(\rho)\ket{r}\right]$ the sum in \eqref{eq:achieving_ent} is equal to
\begin{align}
    \sum_{r=0}^{k-1}&\left|\left|\bra{r}\widetilde{\mathcal{A}_U}(\rho)\ket{r}\otimes \ket{r}\!\!\bra{r}-p_{r\, U}(\rho)\widetilde{t_r(U)}\rho \,\widetilde{t_r(U)}^\dagger\otimes\ket{r}\!\!\bra{r}\right|\right|_{\operatorname{tr}}\nonumber\\
    &=\left|\left|\sum_{r=0}^{k-1}\ket{r}\!\!\bra{r}\widetilde{\mathcal{A}_U}(\rho)\ket{r}\!\!\bra{r}-\sum_{r=0}^{k-1}p_{r\, U}(\rho)\left[\widetilde{t_r(U)}\otimes\ket{r}\right]\rho \left[\widetilde{t_r(U)}^\dagger\otimes\bra{r}\right]\right|\right|_{\operatorname{tr}}\text{.}
\end{align}
Thus, expression \eqref{eq:main_entangled} in the main text is indeed an example of an entangled $k$-task.

The last two definitions are closely related. Applying the principle of deferred measurement to an algorithm of \Cref{def_achieving_rand} will give an algorithm of \Cref{def_achieving_ent} for the entangled version of the same $k$-task. In the opposite direction, adding the measurement with $\Pi_r=I\otimes \ket{r}\!\!\bra{r}$ results in the corresponding $k$-postselection oracle algorithm.

Next we show that the random if clause is possible. By the above argument, the entangled if clause follows. We will show that example \eqref{eq:kappa_cases_appendix} yields an algorithm satisfying \Cref{def_achieving_rand} for $k=m=d$, justifying the tick in \Cref{table_det-random}. Process tomography can obtain classical descriptions $\kappa_r(X)\in U(d)$ that satisfy \Cref{def_tom2}. From this classical description, we can build a new circuit that implements an operator arbitrarily close to $c_r(X):=\ket{0}\!\!\bra{0}\otimes I + \ket{1}\!\!\bra{1}\otimes \kappa_r(X)$. If we know the value of $r$, then this strategy implements a superoperator arbitrarily close to $\mathcal{A}_{r\,U,N}$, where
\begin{equation}
    \mathcal{A}_{r\,U,N}(\rho)=\sum_{\substack{X\in\Omega_N,\\r(X)=r\\}}p_{U,N}(X)c_r(X)\rho c_r(X)^\dagger\text{.}\label{eq:superope_from_classical}
\end{equation}
Moreover, $p_{\text{succ}\,U}(\rho)=1$. Then the expression on the left-hand side of \eqref{eq_eps-achieve_many} is close to
\begin{align}
    \sup_{\substack{\mathcal{H}'\\[0.3em]\rho\in\mathcal{D}(\mathcal{H}\otimes\mathcal{H}')}}\sum_{r=0}^{d-1}&\left|\left|\widetilde{\mathcal{A}_{r\,U,N}}(\rho)-\operatorname{tr}\left[\widetilde{\mathcal{A}_{r\,U,N}}(\rho)\right]\widetilde{c_r(U)}\rho \,\widetilde{c_r(U)}^\dagger\right|\right|_{\operatorname{tr}}\nonumber\\
    &\leq 
    \sup_{\substack{\mathcal{H}'\\[0.3em]\rho\in\mathcal{D}(\mathcal{H}\otimes\mathcal{H}')}}\sum_{X\in\Omega_N}p_{U,N}(X)\left|\left|\widetilde{c_r(X)}\rho \widetilde{c_r(X)}^\dagger-\operatorname{tr}\left[\widetilde{c_r(X)}\rho \widetilde{c_r(X)}^\dagger\right]\widetilde{c_r(U)}\rho \,\widetilde{c_r(U)}^\dagger\right|\right|_{\operatorname{tr}}\nonumber\\
    &\leq 
    \sum_{X\in\Omega_N}p_{U,N}(X)\left|\left|c_r(X)\cdot c_r(X)^\dagger-c_r(U)\cdot \,c_r(U)^\dagger\right|\right|_\diamondsuit\nonumber\\
    &\leq \sum_{X\in\Omega_N}p_{U,N}(X) 2 \left|\left|\kappa_r(X)-\kappa_r(U)\right|\right|_{\operatorname{op}}\nonumber\\
    &\leq 2\epsilon_{N} + 4\delta_N\text{.}
\end{align}
\noindent In the first inequality we took the sum of \eqref{eq:superope_from_classical} and $p_{U,N}(X)$ out of the trace norm. The second inequality holds because $\operatorname{tr}[\widetilde{c_r(X)}\rho \widetilde{c_r(X)}^\dagger]=1$ and because moving the supremum into the sum can only increase the expression. The third inequality uses the upper bound on the diamond norm in terms of the operator norm \cite[Lemma 12.6]{mixed_states_aharonov98} and the equality $\left|\left|Y\otimes Z\right|\right|_{\operatorname{op}}=\left|\left|Y\right|\right|_{\operatorname{op}}\left|\left|Z\right|\right|_{\operatorname{op}}$. The last inequality follows from \Cref{def_tom2}.

This shows that the described tomography strategy $\epsilon$-approximately achieves the random if clause, for any $\epsilon>0$. In particular, it achieves the random $d$-if clause, where $c_{\phi_r}(U)=\ket{0}\!\!\bra{0}\otimes I + \ket{1}\!\!\bra{1}\otimes \kappa_r(U)$ for $\kappa_r$ of \eqref{eq:kappa_cases_appendix}.

\section{Symmetric equations for the determinant and the minors}\label{section_minors}
In the following we first prove the symmetric formula for the determinant of \cref{eq_det}, then we use it to prove \Cref{lem:minors}, the symmetric formula for the minors of \cref{eq_minors}.

\begin{mylemma}
\label{lem:Symmetric formula for det}The determinant of $M\in \text{Mat}_{n\times n}\left(\mathbb{C}\right)$
is given by the formula:
\[
\det\left(M\right)=\frac{1}{n!}\sum_{{\pi},\tau\in S_{n}}\operatorname{sgn}\left(\tau\right)\operatorname{sgn}\left({\pi}\right)\prod_{i=1}^{n}M_{\tau\left(i\right),{\pi}\left(i\right)}.
\]
\end{mylemma}

\begin{proof}
By the usual Leibniz formula, we have $\det\left(M\right)=\sum_{{\pi}\in S_{n}}\operatorname{sgn}\left({\pi}\right)\prod_{i=1}^{n}M_{i,{\pi}\left(i\right)}.$
For each $\tau\in S_{n}$, set $M^{\tau}$ to be the matrix obtained
from $M$ by permuting the rows of $M$ according to $\tau^{-1}$,
i.e. $M_{i,j}^{\tau}=M_{\tau\left(i\right),j}$. On the one hand, by
the properties of the determinant function we have $\det\left(M^{\tau}\right)=\operatorname{sgn}\left(\tau\right)\det\left(M\right)$,
and on the other hand by the Leibniz formula: 
\[
\det\left(M^{\tau}\right)=\sum_{{\pi}\in S_{n}}\operatorname{sgn}\left({\pi}\right)\prod_{i=1}^{n}M_{i,{\pi}\left(i\right)}^{\tau}=\sum_{{\pi}\in S_{n}}\operatorname{sgn}\left({\pi}\right)\prod_{i=1}^{n}M_{\tau\left(i\right),{\pi}\left(i\right)}.
\]
Combining the two expressions for $\det\left(M^{\tau}\right)$ and multiplying
by $\operatorname{sgn}\left(\tau\right)$ gives 
\[
\det\left(M\right)=\sum_{{\pi}\in S_{n}}\operatorname{sgn}\left({\pi}\right)\operatorname{sgn}\left(\tau\right)\prod_{i=1}^{n}M_{\tau\left(i\right),{\pi}\left(i\right)}.
\]
Summing this equality over all $\tau\in S_{n}$ and dividing by $n!$
gives the desired result.
\end{proof}
\begin{mylemma}\label{lem:minors}
The $\left(i,j\right)$ minor of a matrix $M\in \text{Mat}_{n\times n}\left(\mathbb{C}\right)$
is given by 
\[
\det\left(M_{\centernot{i},\centernot{j}}\right)=\frac{\left(-1\right)^{i+j}}{\left(n-1\right)!}\sum_{\substack{{\pi},\tau\in S_{n}\\
{\pi}\left(1\right)=j\\
\tau\left(1\right)=i
}
}\operatorname{sgn}\left(\tau\right)\operatorname{sgn}\left({\pi}\right)\prod_{k=2}^{n}M_{\tau\left(k\right),{\pi}\left(k\right)}.
\]
\end{mylemma}

\begin{proof}
We think of $S_{n-1}$ as the subset of $S_{n}$ of permutations fixing
$n$, i.e. $S_{n-1}=\left\{ \alpha\in S_{n}\mid\alpha\left(n\right)=n\right\} .$
For each $m\in\left\{ i,j\right\} $ set $X_{m}=\left\{ \alpha\in S_{n}\mid\alpha\left(1\right)=m\right\} $,
and let $T_{m}:X_{m}\rightarrow S_{n-1}$ be the mapping $T_{m}\left(\alpha\right)=\left(n\ n-1\ ...\ m\right)\circ\alpha\circ\left(1\ 2\ 3...\ n-1\ n\right)$.
Note that $T_{m}\left(\alpha\right)$ lies in $S_{n-1}$ since it fixes
$n$, that $T_{m}$ is a bijection (since $T_{m}\alpha$ it is obtained
from $\alpha$ by composing on both sides with fixed permutations), and
finally that by the multiplicativity of $\operatorname{sgn}$ we have 
\[
\operatorname{sgn}\left(T_{m}\left(\alpha\right)\right)=\left(-1\right)^{n-m} \operatorname{sgn}\left(\alpha\right)\left(-1\right)^{n-1}=\left(-1\right)^{1-m} \operatorname{sgn}\left(\alpha\right).
\]
One can now verify that for all $k\in\left\{ 1,2,...,n-1\right\} $
and for every $\tau\in X_{i}$ and ${\pi}\in X_{j}$ we have:
\[
\left(M_{\centernot{i},\centernot{j}}\right)_{T_{i}\tau\left(k\right),T_{j}{\pi}\left(k\right)}=M_{\tau\left(k+1\right),{\pi}\left(k+1\right)}.
\]
Indeed, if $\tau\left(k+1\right)<i$ then deleting the $i$th row
from $M$ does not change the index of the $\tau\left(k+1\right)$th
row, and correspondingly $\left(n\ n-1\ ...\ i\right)$ fixes $\tau\left(k+1\right)$
and so $\left(T_{i}\tau\right)\left(k\right)=\tau\left(k+1\right)$.
If $\tau\left(k+1\right)>i$ then then deleting the $i$th row from
$M$ decreases the index of the $\tau\left(k+1\right)$th row by one,
and correspondingly $\left(n\ n-1\ ...\ i\right)$ does the same thing
to $\tau\left(k+1\right)$. A similar argument for ${\pi}$ shows
the validity of the last equality. Now, by \Cref{lem:Symmetric formula for det}:
\begin{eqnarray*}
\det\left(M_{\centernot{i},\centernot{j}}\right) & =&\frac{1}{\left(n-1\right)!}\sum_{{\pi}',\tau'\in S_{n-1}}\operatorname{sgn}\left(\tau'\right)\operatorname{sgn}\left({\pi}'\right)\prod_{k=1}^{n-1}\left(M_{\centernot{i},\centernot{j}}\right)_{\tau'\left(k\right),{\pi}'\left(k\right)}\\
 & =&\frac{1}{\left(n-1\right)!}\sum_{\substack{{\pi}\in X_{i},\tau\in X_{j}}
}\operatorname{sgn}\left(T_{i}\tau\right)\operatorname{sgn}\left(T_{j}{\pi}\right)\prod_{k=1}^{n-1}\left(M_{\centernot{i},\centernot{j}}\right)_{T_{i}\tau\left(k\right),T_{j}{\pi}\left(k\right)}\\
 & =&\frac{1}{\left(n-1\right)!}\sum_{\substack{{\pi}\in X_{i},\tau\in X_{j}}
}\left(-1\right)^{1-i}\operatorname{sgn}\left(\tau\right)\left(-1\right)^{1-j}\operatorname{sgn}\left({\pi}\right)\prod_{k=1}^{n-1}M_{\tau\left(k+1\right),{\pi}\left(k+1\right)}\\
 & =&\frac{\left(-1\right)^{i+j}}{\left(n-1\right)!}\sum_{\substack{{\pi}\in X_{i},\tau\in X_{j}}
}\operatorname{sgn}\left(\tau\right)\operatorname{sgn}\left({\pi}\right)\prod_{k=2}^{n}M_{\tau\left(k\right),{\pi}\left(k\right)}.
\end{eqnarray*}
\end{proof}

\section{Linear-optics solution}\label{section_linear_optics}
First, we briefly review the relevant parts of the linear-optics model as described by Aaronson and Arkhipov \cite{aaronson2011computational}. A $d\times  d$ unitary $U$ specifies what happens to each of the $d$ input modes, collected in the set $S_U$:
\begin{equation*}
    \forall k\in S_U\quad a^{\dagger}_k\mapsto \sum_{j\in S_U} U^{\vphantom{\dagger}}_{jk}a^{\dagger}_j\text{,}
\end{equation*}
where $a^{\dagger}_i$ is the creation operator on the $i$th mode. The unitary acts trivially on other modes, which is captured by the following extension of the matrix
\begin{equation}
    U_{jk}:=\delta_{jk} \quad \text{if $j\notin S_U$ or $k\notin S_U$.}\label{eq_boson_other_modes}
\end{equation} Now for any set of modes $S$ we can describe the action of $U$ by
\begin{equation*}
    \forall k\in S\quad a^{\dagger}_k\mapsto \sum_{j\in S \cup S_U} U^{\vphantom{\dagger}}_{jk}a^{\dagger}_j\text{.}
\end{equation*}
A general quantum state of $n$ particles distributed among $m$ modes is
\begin{equation*}
    \left|\psi\right>_{m, n}=\sum_{\vec{n}=\left(n_1,n_2,\dots n_m\right)\in \Phi_{m, n}}\alpha_{\vec{n}}\frac{\big(a^{\dagger}_1\big)^{n_1}\big(a^{\dagger}_2\big)^{n_2}\dots \left(a^{\dagger}_m\right)^{n_m}}{\sqrt{n_1! n_2! \dots n_m!}}\left|\Omega\right>\text{,}
\end{equation*}
where $\Phi_{m, n}$ is the set of possible partitions of $n$ particles into $m$ slots. In words, a state is a result of taking the vacuum state $\left|\Omega\right>=\left|\Omega\right>_{m, 0}$ and applying to it a degree-$n$ polynomial of $m$ variables -- creation operators. A unitary $U$
 acts on the state as a linear transformation of the polynomial's \emph{variables}. The output is:
 \begin{align}
     U\left[\left|\psi\right>_{|S|, n}\right]=\sum_{\vec{n}=\left(n_1,n_2,\dots n_{|S|}\right)\in \Phi_{|S|, n}}\!\!\!\!\alpha_{\vec{n}}\frac{\left(\sum\limits_{j\in S\cup S_U} \!\!\!\!U^{\vphantom{\dagger}}_{j1}a^{\dagger}_j\right)^{n_1}\left(\sum\limits_{j\in S\cup S_U} \!\!\!\!U^{\vphantom{\dagger}}_{j2}a^{\dagger}_j\right)^{n_2}\!\!\dots \left(\sum\limits_{j\in S\cup S_U} \!\!\!\!U^{\vphantom{\dagger}}_{j|S|}a^{\dagger}_j\right)^{n_{|S|}}}{\sqrt{n_1 !n_2 !\dots n_{|S|} !}} \left|\Omega\right>\label{eq_boson_lifted_unitary}\text{.}
 \end{align}
where $S$ are the input state's modes. In particular, for $k\in S\setminus S_U$, $U$ leaves the $k$th mode of the input state untouched (since then \cref{eq_boson_other_modes} applies). \Cref{eq_boson_lifted_unitary} implies that it also leaves the vacuum state ($n=0$) untouched. These are the two crucial reasons why building an if clause in the linear-optics model is easy.

We define \emph{oracle linear-optics} algorithm to be a sequence of gates $V_0$, $\sigma_1(U)$, $V_1$, $\sigma_2(U), \dots\allowbreak V_{{N}-1}$, $\sigma_{N}(U)$, $V_{N}$; $V_i$ representing fixed gates and $\sigma_i(U)$ query gates as in Fig. 2 in the main text, but the gates now act on a vector state according to \cref{eq_boson_lifted_unitary}.  

Next, we can describe the Mach-Zehnder-interferometer solution (Fig. \ref{fig_mach-zehnder}) of the if clause $(c_\phi, \{id, inv\})$ as a single-query \emph{oracle linear-optics} algorithm. In the quantum-circuit model, a general pure-state input to $c_\phi(U)$ decomposes as $\left|\xi\right>=\alpha_0\left|0\right>\otimes\ket{t_0}+\alpha_1\left|1\right>\otimes\ket{t_1}$, where the first Hilbert space is the control qubit and the second the target qu$d$it. Suppose we encode the control in the photon's polarization. Collect in $S_h$ the modes corresponding to horizontal polarization, $|S_h|=d$, and in $S_v$ the modes corresponding to the vertical, $|S_v|=d$. We write this as the $2d$-mode single-photon state 
\begin{equation*}
    \smash{\big|\psi_{\ket{\xi}}\big>_{2d, 1}=\alpha_0\left|\psi_{\ket{t_0}}\right>_{|S_h|, 1}\left|\Omega\vphantom{_{\ket{t_1}}}\right>_{|S_v|, 0}+\alpha_1\left|\Omega\vphantom{_{\ket{t_1}}}\right>_{|S_h|, 0}\left|\psi_{\ket{t_1}}\right>_{|S_v|, 1}}\text{,}
\end{equation*}
i.e. in the first summand there is one photon on the $d$ horizontal modes and zero photons on the $d$ vertical modes and in the second summand vice versa. In Fig. \ref{fig_mach-zehnder} each polarizing beam splitter corresponds to the unitary $V$ which effectively relabels polarization modes to path modes $h\mapsto \operatorname{lower}$ and $v\mapsto \operatorname{upper}$. Thus, the resulting unitary $V^\dagger U V$ acts nontrivially on the $S_v$ modes only. \Cref{eq_boson_lifted_unitary} implies that 
\begin{align}
    \left(V^\dagger U V\right)\left[\left|\psi_{\left|\xi\right>}\right>_{2d, 1}\right] &=\alpha_0 \left|\psi_{\ket{t_0}}\right>_{|S_h|, 1} U\left[\left|\Omega\vphantom{_{\ket{t_1}}}\right>_{|S_v|, 0}\right]+\alpha_1\left|\Omega\vphantom{_{\ket{t_1}}}\right>_{|S_h|, 0}U\left[\left|\psi_{\ket{t_1}}\right>_{|S_v|, 1}\right]\nonumber\\
    &=\alpha_0\left|\psi_{\ket{t_0}}\right>_{|S_h|, 1}\left|\Omega\vphantom{_{\ket{t_1}}}\right>_{|S_v|, 0}+\alpha_1\left|\Omega\vphantom{_{\ket{t_1}}}\right>_{|S_h|, 0}\left|\psi_{U\ket{t_1}}\right>_{|S_v|, 1}\label{eq_lin}
\end{align}
and the second equality holds because $U$ acts trivially on vacuum, and because $\smash{\left|\psi_{\ket{t_1}}\right>_{|S_v|, 1}}$ has one photon on exactly those modes that are in $S_U$. The output state \eqref{eq_lin} is the linear optics version of
\begin{equation}
    \alpha_0\left|0\right>\otimes\ket{t_0}+\alpha_1\left|1\right>\otimes U\ket{t_1}\text{,}
\end{equation}
which equals $c_\phi(U)\left|\xi\right>$ with $\phi$ identically zero.

\bibliographystyle{unsrt}

\end{document}